\documentclass[journal,romanappendices]{IEEEtran}
\usepackage{mathpple}
\usepackage{times}
\usepackage{amsmath,amsthm}  
\usepackage{amssymb}  
\usepackage[mathscr]{eucal}
\usepackage{cite}     
\usepackage{comment}  
\usepackage{upref}
\usepackage{amsfonts}
\usepackage{verbatim}
\usepackage[usenames,dvipsnames,svgnames,table]{xcolor}
\usepackage{enumerate}
\usepackage{graphicx}
\usepackage{latexsym}
\usepackage{bm}
\usepackage{multirow}
\usepackage{dsfont}
\usepackage{amsthm,xpatch}
\usepackage{mathtools}
\usepackage{bbm}
\usepackage{latexsym}
\usepackage{accents}
\usepackage{eufrak}
\usepackage{psfrag}
\usepackage{color}
\usepackage{times,color}
\usepackage[usenames,dvipsnames,svgnames,table]{xcolor}
\usepackage{pgf}
\usepackage{tikz}
\usetikzlibrary{arrows, automata, positioning, calc}
\usepackage{cancel} 
\usepackage{rotating}
\usepackage[T1]{fontenc}
\usepackage{anyfontsize}
\usepackage[hyphens]{url}
\usepackage{hyperref}
\usepackage[hyphenbreaks]{breakurl}
\usepackage{multirow}
\usepackage{comment}
\usetikzlibrary{fit,matrix}

\newcommand\bovermat[2]{%
  \makebox[0pt][l]{$\smash{\overbrace{\phantom{%
    \begin{matrix}#2\end{matrix}}}^{\text{#1}}}$}#2}

\newcommand\bundermat[2]{%
  \makebox[0pt][l]{$\smash{\underbrace{\phantom{%
    \begin{matrix}#2\end{matrix}}}_{\text{#1}}}$}#2}
\newtheorem{lemma}{Lemma}
\newtheorem{theorem}{Theorem}
\newtheorem{remark}{Remark}

\usepackage{algorithm,algpseudocode}

\makeatletter
\newcommand{\removelatexerror}{\let\@latex@error\@gobble}
\makeatletter
\newcommand{\proofpart}[2]{%
	\par
	\addvspace{\medskipamount}%
	\noindent\emph{Part #1: #2}\par\nobreak
	\addvspace{\smallskipamount}%
	\@afterheading
}
\makeatother

\makeatletter
\newcommand*{\transpose}{%
  {\mathpalette\@transpose{}}%
}
\newcommand*{\@transpose}[2]{%
  \raisebox{\depth}{$\m@th#1\intercal$}%
}
\makeatother

\renewcommand{\mathsf}[1]{#1}

\theoremstyle{definition}
\newtheorem{example}{Example}

\IEEEoverridecommandlockouts

\begin{document}

\newcommand{\SB}[3]{
\sum_{#2 \in #1}\biggl|\overline{X}_{#2}\biggr| #3
\biggl|\bigcap_{#2 \notin #1}\overline{X}_{#2}\biggr|
}

\newcommand{\Mod}[1]{\ (\textup{mod}\ #1)}

\newcommand{\overbar}[1]{\mkern 0mu\overline{\mkern-0mu#1\mkern-8.5mu}\mkern 6mu}

\makeatletter
\newcommand*\nss[3]{%
  \begingroup
  \setbox0\hbox{$\m@th\scriptstyle\cramped{#2}$}%
  \setbox2\hbox{$\m@th\scriptstyle#3$}%
  \dimen@=\fontdimen8\textfont3
  \multiply\dimen@ by 4             % 4x the default rule thickness
  \advance \dimen@ by \ht0
  \advance \dimen@ by -\fontdimen17\textfont2
  \@tempdima=\fontdimen5\textfont2  % x-height
  \multiply\@tempdima by 4
  \divide  \@tempdima by 5          % 80% of the x-height
  % Modifications are only necessary if the top of the subscript is not that high:
  \ifdim\dimen@<\@tempdima
    \ht0=0pt                        % don't let the subscript interfere
    \@tempdima=\fontdimen5\textfont2
    \divide\@tempdima by 4          % 25% of the x-height
    \advance \dimen@ by -\@tempdima % if >0, add to depth of superscript!
    \ifdim\dimen@>0pt
      \@tempdima=\dp2
      \advance\@tempdima by \dimen@
      \dp2=\@tempdima
    \fi
  \fi
  #1_{\box0}^{\box2}%
  \endgroup
  }
\makeatother

\makeatletter
\renewenvironment{proof}[1][\proofname]{\par
  \pushQED{\qed}%
  \normalfont \topsep6\p@\@plus6\p@\relax
  \trivlist
  \item[\hskip\labelsep
        \itshape
%    #1\@addpunct{.}]\ignorespaces% DELETED
    #1\@addpunct{:}]\ignorespaces% ADDED
}{%
  \popQED\endtrivlist\@endpefalse
}
\makeatother

\makeatletter
\newsavebox\myboxA
\newsavebox\myboxB
\newlength\mylenA

\newcommand*\xoverline[2][0.75]{%
    \sbox{\myboxA}{$\m@th#2$}%
    \setbox\myboxB\null% Phantom box
    \ht\myboxB=\ht\myboxA%
    \dp\myboxB=\dp\myboxA%
    \wd\myboxB=#1\wd\myboxA% Scale phantom
    \sbox\myboxB{$\m@th\overline{\copy\myboxB}$}%  Overlined phantom
    \setlength\mylenA{\the\wd\myboxA}%   calc width diff
    \addtolength\mylenA{-\the\wd\myboxB}%
    \ifdim\wd\myboxB<\wd\myboxA%
       \rlap{\hskip 0.5\mylenA\usebox\myboxB}{\usebox\myboxA}%
    \else
        \hskip -0.5\mylenA\rlap{\usebox\myboxA}{\hskip 0.5\mylenA\usebox\myboxB}%
    \fi}
\makeatother

\xpatchcmd{\proof}{\hskip\labelsep}{\hskip3.75\labelsep}{}{}

\pagestyle{plain}

\title{\fontsize{20}{27}\selectfont Multi-Server Private Linear Computation with\\ Joint and Individual Privacy Guarantees}

\author{Nahid Esmati and Anoosheh Heidarzadeh\thanks{The authors are with the Department of Electrical and Computer Engineering, Texas A\&M University, College Station, TX 77843 USA (E-mail: \{nahid, anoosheh\}@tamu.edu).}}

%\thanks{This material is based upon work supported by the National Science Foundation (NSF) under Grants No.~1718658 and 1642983.}

% 

%\thanks{This work was supported by the National Science Foundation under Grant No.~CNS-0954153 and the AFOSR under Contract No.~FA9550-13-1-0008.}

\maketitle

\thispagestyle{plain}

\begin{abstract}
This paper considers the problem of multi-server Private Linear Computation, under the joint and individual privacy guarantees. 
In this problem, identical copies of a dataset comprised of $K$ messages are stored on $N$ non-colluding servers, and a user wishes to obtain one linear combination of a $D$-subset of messages belonging to the dataset. 
The goal is to design a scheme for performing the computation such that the total amount of information downloaded from the servers is minimized, while the privacy of the $D$ messages required for the computation is protected. 
When joint privacy is required, the identities of all of these $D$ messages must be kept private jointly, and when individual privacy is required, the identity of every one of these $D$ messages must be kept private individually. 
In this work, we characterize the capacity, which is defined as the maximum achievable download rate, under both joint and individual privacy requirements. 
In particular, we show that when joint privacy is required the capacity is given by ${(1+1/N+\dots+1/N^{K-D})^{-1}}$, and when individual privacy is required the capacity is given by ${(1+1/N+\dots+1/N^{\lceil K/D\rceil-1})^{-1}}$ assuming that $D$ divides $K$, or $K\pmod D$ divides $D$. 
Our converse proofs are based on reduction from two variants of the multi-server Private Information Retrieval problem in the presence of side information. 
Our achievability schemes build up on our recently proposed schemes for single-server Private Linear Transformation and the multi-server private computation scheme proposed by Sun and Jafar. 
Using similar proof techniques, we also establish upper and lower bounds on the capacity for the cases in which the user wants to compute $L$ (potentially more than one) linear combinations. 
Specifically, we show that when joint privacy is required the capacity is upper bounded by ${(1+1/N+\dots+1/N^{(K-D)/L})^{-1}}$ assuming that $L$ divides ${K-D}$, and lower bounded by ${(1+1/N+\dots+1/N^{K-D+L-1})^{-1}}$. 
\end{abstract}

% ---tailored to the single-server setting---

%The goal of the user is to perform this computation by downloading minimum amount of information from the server, while protecting the privacy of the identities of the $D$ messages required for the computation from the server. 

%In addition, we discuss several extensions of the PLT problem, including PLT in the presence of side information, and     

%\vspace{-0.01cm}

\section{introduction}
In this work, we consider the problem of \emph{multi-server Private Linear Computation} under two different privacy guarantees, referred to as \emph{joint privacy} and \emph{individual privacy}. 
This problem includes $N$ servers storing identical copies of a dataset consisting of $K$ independent and uniformly distributed messages; and 
a user who wishes to compute one linear combination of a set of $D$ messages belonging to the dataset. 
The goal is to design a scheme for performing the computation such that the download rate is maximized (i.e., the total amount of information downloaded from the servers is minimized), while protecting the %(joint or individual)
privacy of identities of the $D$ messages required for the computation. % \emph{jointly} or \emph{individually}. 
%The privacy requirement in JPLC and IPLC is different. 
%joint and individual privacy guarantees, which we refer to as the JPLC and IPLC problems, respectively. 

When joint privacy is required, %The joint privacy requirement implies that 
the identities of all $D$ messages required for the computation must be kept private jointly. 
This notion of privacy is relevant when it is required to hide the correlation between the identities of the messages required for the computation~\cite{HES2021JointJournal}. 
On the other hand, when individual privacy is required, %requirement implies that 
the identity of every one of the $D$ messages required for the computation must be kept private individually. 
%The JPLC problem 
%For instance, the user may want to compute a linear combination of two vectors, and the server must not learn which pair of vectors were required for the computation. 
This notion of privacy, which is a relaxed version of joint privacy, is of practical importance in scenarios in which it is required to hide the information about whether an individual message is used for the computation~\cite{HES2021IndividualJournal}.  
When joint privacy or individual privacy is required, the problem is referred to as \emph{Jointly-Private Linear Computation (JPLC)} or \emph{Individually-Private Linear Computation (IPLC)}, respectively.

The JPLC and IPLC problems in the single-server setting were previously studied in~\cite{HS2019PC} and~\cite{HS2020}, respectively.
These problems are special cases of the Private Linear Transformation (PLT) problem with joint privacy guarantees (JPLT)~\cite{EHS2021JointISIT} and individual privacy guarantees (IPLT)~\cite{EHS2021IndividualISIT}.
In PLT, the user is interested in computing $L$ (potentially more than one) linear combinations of a set of $D$ messages.
The multi-server setting of these problems, however, was not studied previously. 
In this work, we take the first step towards understanding the fundamental limits of multi-server JPLT and IPLT. 
In particular, we characterize the capacity of JPLC and IPLC in the multi-server setting, where the capacity of JPLC (or IPLC) is defined as the maximum achievable download rate over all JPLC (or IPLC) schemes.

The JPLC and IPLC problems are closely related to the Private Linear Computation (PLC) problem, which was originally studied in~\cite{SJ2018,MM2018}. 
(Several variants of the PLC problem were also considered in~\cite{OK2018,OLRK2018,TM2019No2,OLRK2020,YLR2020}.)
In PLC, the values of the combination coefficients in the required linear combination must be kept private. 
This privacy requirement is stronger than those in JPLC and IPLC where only the identities (and not the values of the combination coefficients) of the messages required for the computation need to be kept private. 

%and Private Monomial Computation (PMC)~\cite{YLR2020}
%, and in PMC, the values of the exponents in the required monomial function must be kept private

%The differences between the privacy requirements in these works is that 

% (which is asymptotically equivalent to the multi-server PLC problem as the field size grows unbounded), 
%joint and individual privacy are to protect the data access patterns, and not the values of the coefficients (or the exponents) in the required linear combination (or the required monomial function). 

\subsection{Main Contributions}
In this work, we prove an upper bound on the capacity of JPLC and IPLC by leveraging the existing results on the capacity of two variants of the Private Information Retrieval (PIR) problem. 
We also show the tightness of these bounds by designing JPLC and IPLC schemes that build up on the previously proposed schemes for PLT and PLC.

First, we prove that the capacity of JPLC is given by ${(1+1/N+\dots+1/N^{K-D})^{-1}}$.
To prove this result, we show a reduction from the multi-server PIR with private side information (PIR-PSI) problem~\cite{KGHERS2020} to the JPLC problem. 
In particular, we show that the PIR-PSI problem with $N$ servers, $K$ messages, and $M$ side information messages can be solved by any JPLC scheme for the setting with parameters $N$, $K$, and $D = M+1$.
Using the result of~\cite{CWJ2020} on the capacity of PIR-PSI, we then prove the converse. 
To prove the achievability, we present a JPLC scheme that leverages the single-server JPLT scheme which we recently proposed in~\cite{EHS2021JointISIT} and the multi-server PLC scheme proposed by Sun and Jafar in~\cite{SJ2018}. 
In addition, using similar proof techniques, we establish upper and lower bounds on the capacity of JPLT---which is a generalization of JPLC. 
More specifically, we show that the capacity of JPLT for the setting with parameters $N,K,D,L$ is upper bounded by ${(1+1/N+\dots+1/N^{(K-D)/L})^{-1}}$ assuming that $L\mid (K-D)$, and lower bounded by ${(1+1/N+\dots+1/N^{K-D+L-1})^{-1}}$. 

Next, we prove that the capacity of IPLC is given by ${(1+1/N+\dots+1/N^{\lceil K/D\rceil-1})^{-1}}$ when ${D\mid K}$ or ${K\pmod D \mid D}$. 
We first show that the multi-server PIR with side information (PIR-SI) problem~\cite{KGHERS2020} with $N$ servers, $K$ messages, and $M$ side information messages, can be reduced to the IPLC problem with parameters $N$, $K$, and $D=M+1$. 
Then, we prove the converse by relying on the result of~\cite{LG2020CISS} on the capacity of PIR-SI. 
We prove the achievability by presenting an IPLC scheme that builds on the single-server IPLT scheme which we recently proposed in~\cite{EHS2021IndividualISIT} and the multi-server PLC scheme of~\cite{SJ2018}.

Our results show that JPLC and IPLC can be performed more efficiently than PLC in terms of download rate. 
That is, relaxing the privacy requirement to hide the identities (and not the values of the combination coefficients) of the messages required for the computation can increase the capacity.

Our converse proof techniques suggest that the capacity of different settings of PIR with side information can potentially be instrumental in establishing tight upper bounds on the capacity of PLC (or more generally, PLT) settings in which there is no side information. 
In addition, our achievability schemes suggest that the existing private computation schemes for single-server settings can play a significant role in designing optimal schemes for multi-server settings.  

%multi-server PIR problems can contribute to designing capacity achieving achievability schemes for PLC (and potentially PLT). 

%It can be seen that the results of the single-server PIR-PSI and PIR-SR are instrumental in establishing the upper bounds in the multi-server JPLC and IPLC problems, respectively. This highlights the contribution/significance/role of PIR problems to solve      

%In addition, a simple comparison shows 

%A simple approach is to use multi-message PIR for solving the problem. 
%For both joint and individual privacy guarantees, in single-server setting it has been shown that PLC can be performed more efficiently. We study the advantage of PLC over PIR in multi-server setting. 

\subsection{Notation}
%We denote random variables and their realizations by bold-face and regular symbols, respectively. 
%We denote sets, vectors, and matrices by roman font, %(\texttt{\textbackslash mathrm}), 
%and denote collections of sets, vectors, or matrices by blackboard bold roman font. %(\texttt{\textbackslash mathbbm}). %, and fields are denoted by blackboard bold sans serif font (\texttt{\textbackslash mathbbmss}). 
%For any matrix $\mathrm{M}$, we denote by $\mathrm{rank}(\mathrm{M})$ the rank of the matrix $\mathrm{M}$. 
%For any events $E_1,E_2$, we denote by $\Pr(E_1)$ and $\Pr(E_1|E_2)$ the probability of $E_1$ and the conditional probability of $E_1$ given $E_2$, respectively. 
For any random variables $\mathbf{X},\mathbf{Y}$, $H(\mathbf{X})$ and $H(\mathbf{X}|\mathbf{Y})$ denote the entropy of $\mathbf{X}$ and the conditional entropy of $\mathbf{X}$ given $\mathbf{Y}$, respectively. 
For any integer $n\geq 1$, we denote $\{1,\dots,n\}$ by $[n]$, and for any integers $n<m$, we denote $\{n,n+1,\dots,m\}$ by $[n:m]$.
We denote the binomial coefficient $\binom{n}{k}$ by $C_{n,k}$.  
%Throughout this paper, random variables and their realizations are denoted by bold-face symbols (e.g., $\mathbf{X},\mathbf{W}$) and non-bold-face symbols (e.g., $\mathrm{X},\mathrm{W}$), respectively. 

%\cleardoublepage

\section{Problem Setup}\label{sec:PS}
\subsection{Models and Assumptions}
Let $q$ be an arbitrary prime power, and let $T\geq 1$ be an arbitrary integer. 
Let $\mathbbmss{F}_q$ be a finite field of order $q$, ${\mathbbmss{F}_q^{\times} \triangleq \mathbbmss{F}_q\setminus \{0\}}$ be the multiplicative group of $\mathbbmss{F}_q$, and $\mathbbmss{F}_{q}^{T}$ be the vector space of dimension $T$ over $\mathbbmss{F}_q$.
Let $B\triangleq T\log_2 q$. 
Let $N>1$ be an arbitrary integer, and $K,D\geq 1$ be integers such that ${D\leq K}$. 
%Let ${[K]\triangleq \{1,\dots,K\}}$ be the set of integers from $1$ to $K$. 
We denote by $\mathbbm{W}$ the set of all $D$-subsets (i.e., all subsets of size $D$) of $[K]$. %denote by $\mathscr{V}$ the collection of all $L\times D$ matrices (with entries from $\mathbbmss{F}_p$) with $L$ linearly independent rows and $D$ nonzero columns. 
Also, we denote by $\mathbbm{V}$ the set of all row-vectors of length $D$ with entries in $\mathbbmss{F}_q^{\times}$.

Consider $N$ non-colluding servers each of which stores an identical copy of $K$ messages ${X_1,\dots,X_K}$, where $X_i\in \mathbbmss{F}_q^{T}$ for $i\in [K]$ is a row-vector of length $T$ with entries in $\mathbbmss{F}_q$. 
Let ${\mathrm{X}\triangleq [X_1^{\transpose},\dots,X_K^{\transpose}]^{\transpose}}$. Note that $\mathrm{X}$ is a matrix of size $K\times T$.
For every ${\mathrm{S}\subset [K]}$, we denote by $\mathrm{X}_{\mathrm{S}}$ the submatrix of $\mathrm{X}$ restricted to its rows indexed by $\mathrm{S}$, i.e., $\mathrm{X}_{\mathrm{S}} = [X_{i_1}^{\transpose},\dots,X_{i_{s}}^{\transpose}]^{\transpose}$, where ${\mathrm{S} = \{i_1,\dots,i_{s}\}}$. 
Note that $\mathrm{X}_{\mathrm{S}}$ is a matrix of size $|\mathrm{S}|\times T$, where $|\mathrm{S}|$ denotes the size of $\mathrm{S}$.
Consider a user who wishes to compute one linear combination of $D$ messages, namely, $\mathrm{Z}^{[\mathrm{W},\mathrm{V}]}\triangleq \mathrm{V} \mathrm{X}_{\mathrm{W}}$, where $\mathrm{W}\in \mathbbm{W}$ is the index set of the $D$ messages required for the computation, and $\mathrm{V}\in \mathbbm{V}$ is the coefficient vector of the required linear combination. 
%We represent the collection of the required linear combinations in the matrix form as $\mathrm{Z}^{[\mathrm{W},\mathrm{V}]}\triangleq \mathrm{V}\mathrm{X}_{\mathrm{W}}=\mathrm{U}\mathrm{X}$, where ${\mathrm{V}= [\mathrm{v}_{1}^{\transpose},\dots,\mathrm{v}_{L}^{\transpose}]^{\transpose}}$ is an $L\times D$ matrix with entries in $\mathbbmss{F}_q$, denoting the coefficient matrix pertaining to the required linear combinations, and $\mathrm{U}$ is an $L\times K$ matrix such that the submatrix of $\mathrm{U}$ restricted to the columns indexed by $\mathrm{W}$ is equal to $\mathrm{V}$, and the rest of the columns of $\mathrm{U}$ are all-zero. 
Note that $\mathrm{Z}^{[\mathrm{W},\mathrm{V}]}$ is a row-vector of length $T$ with entries in $\mathbbmss{F}_q$. 
We refer to $\mathrm{Z}^{[\mathrm{W},\mathrm{V}]}$ as the \emph{demand}, $\mathrm{W}$ as the \emph{support of the demand}, $\mathrm{V}$ as the \emph{coefficient vector of the demand}, %$\mathrm{U}$ as the \emph{global coefficient matrix of the demand}, 
and $D$ as the \emph{support size of the demand}. % and $L$ as the \emph{dimension of the demand}.  

In this work, we make the following assumptions: 
\begin{enumerate}
\item $\mathbf{X}_1,\dots,\mathbf{X}_K$ are independent and uniformly distributed over $\mathbbmss{F}_{q}^{T}$. 
%Thus, ${H(\mathbf{X}_{i})=B}$ for $i\in [K]$, where $B\triangleq \log_2 q$, and more generally, 
Thus, $H(\mathbf{X})=KB$, ${H(\mathbf{X}_{\mathrm{S}})= |\mathrm{S}| B}$ for every ${\mathrm{S}\subset [K]}$, and  $H(\mathbf{Z}^{[\mathrm{W},\mathrm{V}]})=B$. % for Model~I and Model~II.
\item $\mathbf{W}, \mathbf{V}, \mathbf{X}$ are independent random variables. 
\item $\mathbf{W}$ is distributed uniformly over ${\mathbbm{W}}$.\footnote{Under the  assumption that $\mathbf{W}$ has a uniform distribution over $\mathbbmss{W}$, it follows that ${\Pr(\mathbf{W}=\tilde{\mathrm{W}})=1/C_{K,D}}$ for all ${\tilde{\mathrm{W}}\in \mathbbmss{W}}$, and ${\Pr(i\in \mathbf{W})}=\sum_{\tilde{\mathrm{W}}\in \mathbbmss{W}: i\in \tilde{\mathrm{W}}} \Pr(\mathbf{W}=\tilde{\mathrm{W}})=C_{K-1,D-1}/C_{K,D}=D/K$ for all ${i\in [K]}$.}  
\item $\mathbf{V}$ is distributed uniformly over ${\mathbbm{V}}$. % that are \emph{Maximum Distance Separable (MDS)}, i.e., every $L\times L$ submatrix of $\mathrm{V}$ is invertible; 
\item The demand's support size $D$ %the underlying model (i.e., whether $\mathbf{V}$ is distributed uniformly over $\mathbbm{V}_{I}$ or $\mathbbm{V}_{I\hspace{-0.04cm}I}$), 
and the distribution of $(\mathbf{W},\mathbf{V})$ are initially known by the server, whereas the realization $(\mathrm{W},\mathrm{V})$ is initially unknown to the server.
\end{enumerate}

\subsection{Privacy and Recoverability Conditions}
For each ${n\in [N]}$, the user generates a query $\mathrm{Q}_n^{[\mathrm{W},\mathrm{V}]}$ given $\mathrm{W}$ and $\mathrm{V}$, simply denoted by $\mathrm{Q}_n$, and sends it to server $n$. 
For simplicity, we denote $\mathbf{Q}_n^{[\mathbf{W},\mathbf{V}]}$ by $\mathbf{Q}_n$. 
Each query $\mathrm{Q}_n$ is a deterministic or stochastic function of $\mathrm{W}$ and $\mathrm{V}$. 
%In the case of a deterministic query, $H(\mathbf{Q}|\mathbf{W},\mathbf{V})=0$, and in the case of a stochastic query, $H(\mathbf{Q}|\mathbf{W},\mathbf{V},\mathbf{R})=0$, where $\mathrm{R}$ is a random key generated by the user (independently from $\mathrm{W},\mathrm{V},\mathrm{X}$), and unknown to the servers. 
%initially generated by the user---independently from $\mathrm{W},\mathrm{V},\mathrm{X}$, and initially unknown to the servers.
%Here, $\mathbf{Q}$ denotes $\mathbf{Q}^{[\mathbf{W},\mathbf{V}]}$. %where $\mathbf{Q}^{[\mathbf{W},\mathbf{V}]}$ is denoted by $\mathbf{Q}$.
%, and potentially a random key $\mathrm{R}$ that is generated by the user in advance---independently from $\mathrm{W},\mathrm{V},\mathrm{X}$, and is initially unknown to the servers. 
%That is, $H(\mathbf{Q}|\mathbf{W},\mathbf{V},\mathbf{R})=0$, where $\mathbf{Q}^{[\mathbf{W},\mathbf{V}]}$ is denoted by $\mathbf{Q}$. 

The queries $\mathrm{Q}_n$'s must satisfy a privacy condition. 
In this work, we consider two different privacy conditions: 
\begin{itemize}
    \item \emph{Joint Privacy:} Given the query $\mathrm{Q}_n$, every $D$-subset of message indices must be equally likely to be the demand's support $\mathbf{W}$ from the perspective of server $n$, i.e., for every $\tilde{\mathrm{W}}\in \mathbbm{W}$, for all $n\in [N]$ it must hold that
\begin{equation*}
\Pr (\mathbf{W}=\tilde{\mathrm{W}}|\mathbf{Q}_n=\mathrm{Q}_n)=\Pr(\mathbf{W}=\tilde{\mathrm{W}}). 
\end{equation*} 
    \item \emph{Individual Privacy:} Given the query $\mathrm{Q}_n$, every message index must be equally likely to belong to the demand's support $\mathbf{W}$ from the perspective of the server $n$, i.e., for every ${i\in [K]}$, for all $n\in [N]$ it must hold that
\begin{equation*}
\Pr (i\in \mathbf{W}|\mathbf{Q}_n=\mathrm{Q}_n)=\Pr(i\in \mathbf{W}). 
\end{equation*}
\end{itemize} 
%Note that $\Pr(\mathbf{W}=\tilde{\mathrm{W}}) = 1/C_{K,D}$ for all $\tilde{\mathrm{W}}\in \mathbbmss{W}$, and $\Pr(i\in \mathbf{W}) = D/K$ for all $i\in [K]$.
% ={1}/{C_{K,D}}
% =D/K
%We refer to this condition as the \emph{joint privacy condition}. 

Upon receiving the query $\mathrm{Q}_n$, server $n$ generates an answer $\mathrm{A}_n^{[\mathrm{W},\mathrm{V}]}$, simply denoted by $\mathrm{A}_n$, and sends it back to the user. 
For simplicity, we denote $\mathbf{A}_n^{[\mathbf{W},\mathbf{V}]}$ by $\mathbf{A}_n$.
The answer $\mathrm{A}_n$ is a deterministic function of $\mathrm{Q}_n$ and $\mathrm{X}$.
That is, ${H(\mathbf{A}_n|\mathbf{Q}_n,\mathbf{X})=0}$. %, where $\mathbf{A}^{[\mathbf{W},\mathbf{V}]}$ is denoted by $\mathbf{A}$. 

The answers $\mathrm{A}_{[N]}\triangleq\{\mathrm{A}_n\}_{n\in [N]}$, the queries $\mathrm{Q}_{[N]}\triangleq\{\mathrm{Q}_n\}_{n\in [N]}$, and the realization $(\mathrm{W},\mathrm{V})$ must collectively enable the user to retrieve the demand
$\mathrm{Z}^{[\mathrm{W},\mathrm{V}]}$, i.e., 
\[H(\mathbf{Z}| \mathbf{A}_{[N]},\mathbf{Q}_{[N]}, \mathbf{W},\mathbf{V})=0,\] where 
$\mathbf{Z}^{[\mathbf{W},\mathbf{V}]}$ is denoted by $\mathbf{Z}$ for the ease of notation. %, $\mathbf{A}_{[N]}\triangleq \{\mathbf{A}_n\}_{n\in [N]}$, and $\mathbf{Q}_{[N]}\triangleq \{\mathbf{Q}_n\}_{n\in [N]}$. 
We refer to this condition as the \emph{recoverability condition}. %Note that $H(\mathbf{Z})=LB$.

\subsection{Problem Statement}
The problem is to design a protocol for generating a collection of queries $\mathrm{Q}_{1}^{[\mathrm{W},\mathrm{V}]},\dots,\mathrm{Q}_{N}^{[\mathrm{W},\mathrm{V}]}$ and the corresponding answers $\mathrm{A}_{1}^{[\mathrm{W},\mathrm{V}]},\dots,\mathrm{A}_{N}^{[\mathrm{W},\mathrm{V}]}$ for any given $\mathrm{W}$ and $\mathrm{V}$
such that the privacy and recoverability conditions are satisfied.
We refer to this problem as \emph{Jointly-Private Linear Computation (JPLC)} or \emph{Individually-Private Linear Computation (IPLC)} when joint or individual privacy is required, respectively.
%A protocol for MJPLC or MIPLC is referred to as a \emph{MJPLC} or \emph{MIPLC} \emph{protocol}. 
%A protocol is called \emph{linear} if the server's answer to the user's query consists only of linear combinations of the messages; otherwise, the protocol is called \emph{non-linear}.  

%Following the convention in the PIR and PLC literature, 
The \emph{rate} of a JPLC or IPLC protocol is defined as the ratio of the entropy of the demand (i.e., $H(\mathbf{Z})=B$) to the total entropy of the answers (i.e., $H(\mathbf{A}_{[N]})$). 
We define the \emph{capacity} of JPLC (or IPLC) setting as the supremum of rates over all JPLC (or IPLC) protocols, all field sizes $q$, and all message lengths $T$.
%It should be noted that the capacity may depend on $q$. 
%However, in this work, we are interested in the supremum of rates over all protocols and all $q$. 
In this work, our goal is to characterize %derive (tight) lower and upper bounds on 
the capacity of JPLC and IPLC settings in terms of $N,K,D$. 
 
%In this work we would like to characterize the supremum of rates over all protocols and all $q$. %\footnote{Our converse bounds hold for any $q$, and our achievability schemes achieve these converse bounds when $q$ is sufficiently large, depending on $K,D,L$.}

\section{Main Results}
This section present our main results.  
Theorems~\ref{thm:MSJPLC} and~\ref{thm:MSIPLC} characterize the capacity of JPLC and IPLC settings, respectively. 
The proofs are given in Sections~\ref{sec:MSJPLC} and~\ref{sec:MSIPLC}, respectively. 

\begin{theorem}\label{thm:MSJPLC}
For the JPLC setting with $N$ servers, $K$ messages, and demand's support size $D$, the capacity is given by 
\begin{equation}\label{eq:JPLCCap}
\left(1+\frac{1}{N}+\frac{1}{N^2}+\dots+\frac{1}{N^{K-D}}\right)^{-1}.    
\end{equation} 	
\end{theorem}

To prove the converse, we show that the PIR-PSI problem can be reduced to the multi-server JPLC problem. 
More specifically, we show that the problem of PIR-PSI with $N$ servers, $K$ messages, and $M$ messages as side information can be solved using any JPLC protocol for $N$ servers, $K$ messages, and demand's support size $D = M+1$.
Using this reduction, we prove the converse bound on the capacity of JPLC based on the result of~\cite{CWJ2020} on the capacity of PIR-PSI. 
We prove the achievability result by presenting a multi-server JPLC protocol that achieves the converse bound. 
This protocol builds up on the single-server JPLT scheme of~\cite{EHS2021JointISIT}, and the multi-server PLC scheme of~\cite{SJ2018}. 
The single-server JPLT scheme is used for constructing the smallest possible set of linear combinations of the messages, which we refer to as coded messages, 
that satisfies the following two requirements:
\begin{itemize}
    \item[(i)] For every $D$-subset of messages, there is a linear combination of these $D$ messages which can be obtained by linearly combining the coded messages.
    \item[(ii)] The linear combination required by the user can be obtained by linearly combining the coded messages. 
\end{itemize}
The multi-server PLC scheme is then used to retrieve the linear combination required by the user, while not revealing which linear combination of coded messages is retrieved.

\begin{remark}\label{rem:JPLC1}
\emph{Theorem~\ref{thm:MSJPLC} extends the result of~\cite[Theorem~1]{EHS2021JointISIT} for single-server JPLC as a special case of single-server JPLT~\cite{EHS2021JointISIT}. 
In JPLT, the user wishes to compute $L$ linear combinations of a $D$-subset of $K$ messages while hiding the index set of the $D$ messages required for the computation. 
As shown in~\cite{EHS2021JointISIT}, the capacity of single-server JPLT is given by $L/(K-D+L)$, and hence, for the single-server setting the capacity of JPLC, which is equivalent to JPLT for $L=1$, is given by $1/(K-D+1)$. 
This result matches the result of Theorem~\ref{thm:MSJPLC} when $N=1$.} 
\end{remark}

\begin{remark}
\normalfont 
Using the same technique as in our converse proof of Theorem~\ref{thm:MSJPLC}, we can show a reduction from the multi-server Multi-Message PIR-PSI (MPIR-PSI) problem~\cite{SSM2018} to the multi-server extension of the JPLT problem introduced in~\cite{EHS2021JointISIT}---which is a generalization of the multi-server JPLC problem considered in this work. 
In particular, we can show that the MPIR-PSI problem with $N$ servers, $K$ messages, $P$ demand messages, and $M$ side information messages can be solved by any multi-server JPLT protocol for the setting in which there are $N$ servers and $K$ messages, and the user wants to compute $L=P$ linear combinations of a set of $D=P+M$ messages.
Using this reduction and the result of~\cite{SSM2018} on the capacity of MPIR-PSI, it is easy to show that the rate of any multi-server JPLT protocol for the setting with parameters $N,K,D,L$ is upper bounded by ${(1+1/N+\dots+1/N^{(K-D)/L})^{-1}}$ when $L\mid (K-D)$. 
The special case of this result for $L=1$ matches the converse bound for the multi-server JPLC problem. 
The tightness of this bound, however, remains unknown in general. 
%matches the result of Lemma~\ref{lem:JPLC-Conv}. 
Using a similar idea as in our achievability scheme for multi-server JPLC, one can design a scheme for multi-server JPLT. 
%The idea is to first use the single-server JPLT scheme of~\cite{EHS2021JointISIT} for constructing $K-D+L$ coded messages, and then use the multi-server PLC scheme of~\cite{SJ2018} $L$ times for separately recovering the $L$ required coded combinations. 
The idea is to utilize the single-server JPLT scheme of~\cite{EHS2021JointISIT} to construct $K-D+L$ coded messages, and then retrieve each of the $L$ desired coded combinations separately by using the multi-server PLC scheme of~\cite{SJ2018}. 
The rate of of this scheme is given by ${(1+1/N+\dots+1/N^{K-D+L-1})^{-1}}$, which does not match our converse bound for multi-server JPLT in general. 
Proving a tight converse bound and designing an optimal scheme for multi-server JPLT are the focus of an ongoing work. 
% However, this scheme does not achieve the converse bound in general. 
%---which is based on combining the single-server JPLC scheme of~\cite{EHS2021JointISIT} and the multi-server PLC scheme of~\cite{SJ2018}--- 
%Unlike our achievability scheme for multi-server JPLC, the converse bound for the multi-server JPLT cannot be achieved by constructing $K-D+L$ coded messages using the single-server JPLT scheme of~\cite{EHS2021JointISIT} and then using the multi-server PLC scheme of~\cite{SJ2018} multiple times to separately recover the $L$ linear combinations of coded messages that are required by the user. 
\end{remark}

\begin{remark}\label{rem:JPLC2}
\emph{The JPLC problem is closely related to the PLC problem. 
In PLC, there is a collection of $M$ linear combinations on $K$ messages, and the user wishes to compute one of these linear combinations while not revealing the identity of the required linear combination. 
In~\cite[Theorem~1]{SJ2018}, it was shown that the capacity of PLC is given by ${(1+1/N+1/N^2+\dots+1/N^{P-1})^{-1}}$, where $P$ is the maximum number of linearly independent combinations within the collection of $M$ linear combinations. 
When this collection consists of all linear combinations with support of size $D$ (for all possible coefficient vectors), it is easy to see that $P=K$. 
Thus, the capacity of PLC for this setting is given by $(1+1/N+1/N^2+\dots+1/N^{K-1})^{-1}$, which does not depend on $D$. 
This is in contrast to the result of Theorem~\ref{thm:MSJPLC}, because the capacity of JPLC increases as $D$ increases.
This implies that 
%A comparison of the results of Theorem~\ref{thm:MSJPLC} for JPLC and~\cite[Theorem~1]{SJ2018} for PLC shows that 
relaxing the privacy condition to hide only the support (and not the values of the combination coefficients) of the required linear combination can increase the capacity.}
\end{remark}

\begin{theorem}\label{thm:MSIPLC}
For the IPLC setting with $N$ servers, $K$ messages, and demand's support size $D$, the capacity is given by
\begin{equation}\label{eq:IPLCCap}
\left(1+\frac{1}{N}+\frac{1}{N^2}+\dots+\frac{1}{N^{\lceil K/D\rceil-1}}\right)^{-1},
\end{equation} if $R=0$ or $R\mid D$, where $R \triangleq K\pmod D$.\end{theorem}

We prove the converse by showing that the PIR-SI problem is reducible to the multi-server IPLC problem. 
To be more specific, we prove the converse bound by relying on the result of~\cite{LG2020CISS} for the capacity of multi-server PIR-SI, and showing that any IPLC protocol for $N$ servers, $K$ messages, and demand's support size $D$ can be used for solving the problem of PIR-SI with $N$ servers, $K$ messages, and $M = D-1$ messages as side information. 
To prove the achievability of the converse bound, we propose a multi-server IPLC protocol that builds up on the single-server IPLT scheme of~\cite{EHS2021IndividualISIT} and the multi-server PLC scheme of~\cite{SJ2018}. 
The single-server IPLT scheme is used for constructing the smallest possible set of coded messages (i.e., linear combinations of the messages) and a probability distribution associated with linear combinations of coded messages, which we refer to as coded combinations, that satisfy the following two requirements: 
\begin{itemize}
    \item[(i)] For all $i\in[K]$, the sum of probabilities associated with all coded combinations whose support has size $D$ and contains the message index $i$ is the same.
    \item[(ii)] The linear combination required by the user is one of the coded combinations. 
\end{itemize}
The multi-server PLC scheme is then utilized to privately retrieve the linear combination required by the user. % without revealing which coded combination is retrieved.

%The multi-server PLC scheme is then used for retrieving the user's required linear combination, while leaking no information about which one of linear combinations of coded messages is retrieved.

%The single-server IPLT scheme is used for creating the smallest possible set of coded messages (constructed by linearly combining the messages) and a probability distribution associated with linear combinations of coded messages (referred to as coded combinations) that satisfy the following two requirements:

\begin{remark}\label{rem:IPLC1} 
\emph{The single-server IPLC problem is a special case of the single-server IPLT problem~\cite{EHS2021IndividualISIT}.
Similar to JPLT, in IPLT the user wishes to compute $L$ linear combinations of a $D$-subset of $K$ messages. 
The privacy condition in IPLT is, however, a relaxed version of that in JPLT. 
In particular, unlike JPLT, in IPLT the identities of the $D$ messages required for the computation do not need to be protected jointly; instead, the identity of each of these messages must be kept private individually. 
As shown in~\cite[Theorem~1]{EHS2021IndividualISIT}, when $R\leq L$ or $R\mid D$, the capacity of single-server IPLT is given by $(\lfloor K/D\rfloor + \min\{1,R/L\})^{-1}$. 
Specializing this result for $L=1$, when $R=0$ or $R\mid D$, the capacity of single-server IPLC is given by $(\lfloor K/D\rfloor + \min\{1,R\})^{-1}$, or equivalently, $\lceil K/D\rceil^{-1}$, matching the result of Theorem~\ref{thm:MSIPLC} for $N=1$. 
Using our proof techniques and the result of~\cite[Theorem~1]{EHS2021IndividualISIT} for $L=1$, it is easy to derive lower and upper bounds on the capacity of multi-server IPLC for the settings in which $R\neq 0$ and $R\nmid D$.
Notwithstanding, the capacity of both single-server and multi-server IPLC remains open for these settings.}
\end{remark}

\begin{remark}\label{rem:IPLC2}
\emph{Comparing the results of Theorems~\ref{thm:MSJPLC} and~\ref{thm:MSIPLC}, it is obvious that IPLC can be performed more efficiently than JPLC in terms of download rate. 
The advantage of IPLC over JPLC for fixed $K$ and $D$ is more pronounced for smaller $N$; and as $N$ grows unbounded, the capacity of both JPLC and IPLC converges to $1$. 
In addition, for a fixed $N$, when $K$ and $D$ grow unbounded at the same speed (i.e., $K/D$ is fixed), the capacity of JPLC converges to $1-1/N$, whereas the capacity of IPLC remains constant and is greater than the asymptotic capacity of JPLC by a factor of $1/(1-1/N^{\lceil K/D\rceil})$.}
\end{remark}

\section{Proof of Theorem~\ref{thm:MSJPLC}}\label{sec:MSJPLC}

\subsection{Converse Proof}\label{subsec:MSJPLC-Conv}
In this section, we prove the converse part of Theorem~\ref{thm:MSJPLC}, by upper bounding the rate of JPLC protocols in terms of the parameters $N,K,D$. 
The upper bound holds for any field size $q$ and any message length $T$.

\begin{lemma}\label{lem:JPLC-Conv}
The rate of any JPLC protocol for the setting with parameters $N,K,D$ is upper bounded by~\eqref{eq:JPLCCap}.
%\begin{equation*}
%\left(1+\frac{1}{N}+\frac{1}{N^2}+\dots+\frac{1}{N^{K-D}}\right)^{-1}.    
%\end{equation*}
\end{lemma}

\begin{proof}
To prove the lemma, we show a reduction from the PIR-PSI problem~\cite{KGHERS2020} to the JPLC problem. 
In the PIR-PSI problem with parameters $N,K,M$, 
there are $N$ non-colluding servers each of which stores an identical copy of $K$ messages $X_1,\dots,X_{K}$ (independent and uniformly distributed over $\mathbb{F}_q^{T}$), and 
there is a user who initially knows $M$ of these messages $X_{i_1},\dots,X_{i_M}$ as side information, 
but the support of the user's side information, $\mathrm{S} =\{i_1,\dots,i_M\}$, is not initially known by any of the servers. 
The user wants to retrieve the message $X_{i^{*}}$ for some $i^{*}\in [K]\setminus \mathrm{S}$, where $i^{*}$ is not initially known by any of the servers. 
The goal of the user is retrieve their desired message from the servers with maximum possible download rate, while hiding both the index of the desired message, $i^{*}$, and the support of the side information, $\mathrm{S}$, from any of the servers. 
We refer to $X_{i^{*}}$ as the ``uncoded demand''. 
As shown in~\cite{CWJ2020}, the capacity of PIR-PSI---defined as the maximum achievable download rate---is given by 
\begin{equation}\label{eq:1}
\left(1+1/N+1/N^2+\dots+1/N^{K-M-1}\right)^{-1},    
\end{equation}
under the assumptions that 
$\mathbf{S}$ is distributed uniformly over all $M$-subsets of $[K]$, and 
$\boldsymbol{i}^{*}$ given $\mathbf{S}=\mathrm{S}$ has a uniform distribution over $[K]\setminus \mathrm{S}$.

To show a reduction from PIR-PSI to JPLC, we need to prove that the PIR-PSI problem with parameters $N,K,M$ can be solved by any JPLC protocol for the setting with $N$ servers, $K$ messages, and demand's support size $D=M+1$. 
%Using this reduction, the lemma can be proven by the way of contradiction. 
Once this reduction is established, the proof of the lemma is straightforward by the way of contradiction. 
Suppose that the rate of a JPLC protocol for the setting with parameters $N$, $K$, and $D=M+1$, is higher than ${1/(1+1/N+\dots+1/N^{K-D})}$, or equivalently, ${1/(1+1/N+\dots+1/N^{K-M-1})}$. 
Solving the PIR-PSI problem via this JPLC protocol, one can then achieve a higher rate than the capacity of PIR-PSI given in~\eqref{eq:1}, which is obviously a contradiction.
%Below we show the reduction. %, we proceed as follows. 

Consider the PIR-PSI problem with parameters $N,K,M$. 
To show a reduction, suppose that the user employs an arbitrary JPLC protocol for the setting with parameters $N$, $K$, and $D=M+1$, so as to compute the linear combination ${v_1X_{i^{*}}+v_2X_{i_1}+\dots+v_{M}X_{i_{M}}}$, 
where $i^{*}$ is the index of the user's uncoded demand, ${\mathrm{S} = \{i_1,\dots,i_{M}\}}$ is the support of the user's side information, and $v_1,\dots,v_{M}$ are randomly chosen from $\mathbb{F}_q^{\times}$.
We refer to this linear combination as the  ``coded demand''. 
Note that the support and the coefficient vector of the coded demand are given by $\mathrm{W}=\{i^{*}\}\cup \mathrm{S}$ and $\mathrm{V}=[v_1,v_2,\dots,v_M]$, respectively.
Moreover, $\mathrm{W}\in \mathbbmss{W}$ since $\mathrm{W}\subseteq [K]$ and $|\mathrm{W}|=M+1 = D$, and $\mathrm{V}\in \mathbbmss{V}$ by construction. 
It should also be noted that $\mathbf{W} = \{\boldsymbol{i}^{*}\}\cup \mathbf{S}$ is distributed uniformly over $\mathbbmss{W}$. 
This is because $\mathbf{S}$ is distributed uniformly over all $M$-subsets of $[K]$, and $\boldsymbol{i}^{*}$ given $\mathbf{S}=\mathrm{S}$ is distributed uniformly over $[K]\setminus \mathrm{S}$. 
Moreover, $\mathbf{V}$ has a uniform distribution over $\mathbbmss{V}$, by construction.
By these arguments, it should be obvious that this setting matches the JPLC setting defined in Section~\ref{sec:PS}. 
Using an arbitrary JPLC protocol for the setting with parameters $N,K,D=M+1$, for each $n\in [N]$, the user then generates a query $\mathrm{Q}^{[\mathrm{W},\mathrm{V}]}_n$, and sends it to server $n$, and server $n$ sends back the corresponding answer $\mathrm{A}^{[\mathrm{W},\mathrm{V}]}_n$ to the user. 
For simplifying the notation, we denote $\mathrm{Q}^{[\mathrm{W},\mathrm{V}]}_n$ and $\mathrm{A}^{[\mathrm{W},\mathrm{V}]}_n$ by  $\mathrm{Q}_n$ and $\mathrm{A}_n$, respectively. 

To complete the proof of reduction, we need to show that %the JPLC-based scheme satisfies 
the recoverability and privacy conditions of the PIR-PSI problem are satisfied. 
The recoverability of the user's coded demand is guaranteed since any JPLC protocol satisfies the recoverability condition.
Provided that the user can recover their coded demand ${v_1X_{i^{*}}+v_2X_{i_1}+\dots+v_{M}X_{i_M}}$, it is immediate that the user can recover their uncoded demand $\mathrm{X}_{i^{*}}$. %(Note that %the side information messages
%$X_{i_1},\dots,X_{i_M}$ are known by the user.) 
Thus, the recoverability of the user's uncoded demand is guaranteed.  
To prove that the privacy of both the index of the uncoded demand and the support of the side information is protected, we need to show that, from the perspective of each server, (i) all $\tilde{\mathrm{W}}\in \mathbbmss{W}$ are equally likely to be the union of the index of the uncoded demand and the support of the side information, and (ii) for every $\tilde{\mathrm{W}}\in\mathbbmss{W}$, provided that $\tilde{\mathrm{W}}$ is the support of the coded demand, all $i\in \tilde{\mathrm{W}}$ are equally likely to be the index of the uncoded demand. 

Since any JPLC protocol satisfies the joint privacy condition, ${\Pr(\mathbf{W} = \tilde{\mathrm{W}}|\mathbf{Q}_n = \mathrm{Q}_n) = \Pr(\mathbf{W}=\tilde{\mathrm{W}})}$ for every ${\tilde{\mathrm{W}}\in \mathbbmss{W}}$.  
Since $\mathbf{W} = \{\boldsymbol{i}^{*}\}\cup\mathbf{S}$, % and $\mathbbmss{W}$ is the set of all $D$-subsets of $[K]$, 
it follows that the privacy requirement (i) is satisfied. 
Moreover, for every ${\tilde{\mathrm{W}}\in \mathbbmss{W}}$, given that ${\mathbf{W}=\tilde{\mathrm{W}}}$, the two events $\mathbf{Q}_n = \mathrm{Q}_n$ and $\boldsymbol{i}^{*} = i$ for $i\in \tilde{\mathrm{W}}$ are independent. 
This is because $\mathbf{Q}_n$ is a function of $\mathbf{W}$ and $\mathbf{V}$ (and potentially a random key, independent from $\mathbf{W}$ and $\mathbf{V}$), and $\mathbf{Q}_n$ given $\mathbf{W}$ is independent of $\boldsymbol{i}^{*}$. 
Thus, ${\Pr(\boldsymbol{i}^{*}=i|\mathbf{Q}_n = \mathrm{Q}_n,\mathbf{W} = \tilde{\mathrm{W}}) = \Pr(\boldsymbol{i}^{*}=i|\mathbf{W} = \tilde{\mathrm{W}})}$ for all $i\in \tilde{\mathrm{W}}$.
Since ${\Pr(\boldsymbol{i}^{*}=i)}=1/K$ for all $i\in [K]$, 
${\Pr(\mathbf{S}=\tilde{\mathrm{W}}\setminus\{i\}|\boldsymbol{i}^{*}=i)}=1/C_{K-1,D-1}$ for all $i\in \tilde{\mathrm{W}}$, and 
${\Pr(\mathbf{W}=\tilde{\mathrm{W}})}=1/C_{K,D}$ for all $\tilde{\mathrm{W}}\in \mathbbmss{W}$, % and 
%${\Pr(\boldsymbol{i}^{*}=i,\mathbf{W}=\tilde{\mathrm{W}})} = {\Pr(\boldsymbol{i}^{*}=i,\mathbf{S}=\tilde{\mathrm{W}}\setminus \{i\})}$ for all $i\in \tilde{\mathrm{W}}$, 
it is easy to verify that ${\Pr(\boldsymbol{i}^{*}=i|\mathbf{W} = \tilde{\mathrm{W}}) = 1/D}$ for all $i\in \tilde{\mathrm{W}}$. 
%Moreover, 
%\begin{align*}
%\Pr(\boldsymbol{i}^{*}=i|\mathbf{W} = \tilde{\mathrm{W}}) & = \frac{\Pr(\boldsymbol{i}^{*}=i,\mathbf{W}=\tilde{\mathrm{W}})}{\Pr(\mathbf{W}=\tilde{\mathrm{W}})}\\
%& = \frac{\Pr(\boldsymbol{i}^{*}=i,\mathbf{S}=\tilde{\mathrm{W}}\setminus \{i\})}{\Pr(\mathbf{W}=\tilde{\mathrm{W}})}\\
%& = \frac{1}{K}\times \frac{C_{K,D}}{C_{K-1,D-1}}\\
%& = \frac{1}{D}.
%\end{align*}
Thus, ${\Pr(\boldsymbol{i}^{*}=i|\mathbf{Q}_n = \mathrm{Q}_n,\mathbf{W} = \tilde{\mathrm{W}})=1/D}$ for all $i\in \tilde{\mathrm{W}}$. 
This implies that the privacy requirement (ii) is satisfied. 
%This completes the proof. 
\end{proof}

% \footnote{GRS stands for Generalized Reed-Solomon.}

\subsection{Achievability Scheme}\label{subsec:MSJPLC-Ach}
In this section, we present a JPLC protocol, termed \emph{Multi-Server Specialized GRS Code}, for all parameters $N,K,D$. (Here, GRS stands for Generalized Reed-Solomon.) 
This protocol is capacity-achieving for any field size $q\geq K$ and any message length $T$ that is an integer multiple of $N^{M}$, where $M\triangleq C_{K,D}$.
An example of this protocol is given in Section~\ref{sec:EX}. 

The Multi-Server Specialized GRS Code protocol consists of three steps as described below.\vspace{0.125cm} 

\textbf{Step 1:} 
First, the user constructs a $J\times K$ matrix $\mathrm{G}$ which generates a specific ${[K,J]}$ GRS code, where ${J\triangleq K-D+1}$, by utilizing the single-server JPLT protocol of~\cite{EHS2021JointISIT} for the special case in which one linear combination is required by the user. 
To avoid repetition, we omit the steps of this protocol, and only present the matrix $\mathrm{G}$ being constructed.

%First, the user utilizes the single-server JPLT protocol which we proposed in~\cite{EHS2021JointISIT}, for the special case in which one linear combination is required by the user, in order to construct a $J\times K$ matrix $\mathrm{G}$ which generates a specific ${[K,J]}$ GRS code, where $J\triangleq K-D+1$. 
%To avoid repetition, we omit the steps of this protocol, and only present the matrix $\mathrm{G}$ being constructed using this protocol.     

Recall that $\mathrm{W}$ and $\mathrm{V}$ denote the support and the coefficient vector of the user's demand $\mathrm{Z}$. 
Suppose ${\mathrm{W}=\{i_1,\dots,i_D\}}$, ${[K]\setminus \mathrm{W} = \{i_{D+1},\dots,i_K\}}$, and ${\mathrm{V} = [v_1,\dots,v_D]}$. 
Let $\pi$ be a permutation on $[K]$ such that $\pi(j) = i_j$ for $j\in [K]$.
Let $\omega_1,\dots,\omega_K$ be $K$ arbitrary distinct elements from $\mathbb{F}_q$, and 
let ${v_{D+1},\dots,v_{K}}$ be ${K-D}$ randomly chosen (with replacement) elements from $\mathbbmss{F}_q^{\times}$. 

For every ${i\in [J]}$ and every ${j\in [K]}$, the entry ${(i,\pi(j))}$ of the matrix $\mathrm{G}$ is given by ${\alpha_j \omega_j^{i-1}}$, where 
\begin{equation*}
\alpha_j\triangleq 
\begin{cases}\displaystyle
v_j\prod_{k\in [D+1:K]} (\omega_j-\omega_k)^{-1}, & j\in [D],\\ \displaystyle
v_j \prod_{k\in [K]\setminus \{j\}} (\omega_j-\omega_k)^{-1}, & j\in [D+1:K].
\end{cases}
\end{equation*}
%For every $i\in [J]$, we denote by  $\mathrm{G}_i$ the $i$th row of $\mathrm{G}$. 
Note that the matrix $\mathrm{G}$ generates a $[K,J]$ GRS code, and $\{\alpha_{j}\}_{j\in [K]}$ and $\{\omega_{j}\}_{j\in [K]}$ are the multipliers and the evaluation points of the GRS code generated by $\mathrm{G}$, respectively. 

% (up to scalar multiplication) 

Let $\mathrm{W}_1,\dots,\mathrm{W}_{M}$ be an arbitrary ordering of the elements in $\mathbbmss{W}$. 
As shown in~\cite{EHS2021JointISIT}, the matrix $\mathrm{G}$ has the following properties:
\begin{itemize}
    \item[(i)] For every $k\in [M]$, the row space of $\mathrm{G}$ contains a unique row-vector $\mathrm{U}_k$ with support $\mathrm{W}_k$ and first nonzero coordinate equal to the first coordinate of $\mathrm{V}$ (i.e., $v_1$).
    \item[(ii)] There exists a unique $k^{*}\in [M]$ such that $\mathrm{W}_{k^{*}}=\mathrm{W}$.
    \item[(iii)] The row-vector $\mathrm{U}_{k^{*}}$ with support $\mathrm{W}_{k^{*}}=\mathrm{W}$, when restricted to its nonzero coordinates, is equal to $\mathrm{V}$.
\end{itemize}

By (i), for every $k\in [M]$, the row-vector $\mathrm{U}_k$ can be written as a unique linear combination of the rows of $\mathrm{G}$. 
Since the rows of $\mathrm{G}$ are linearly independent (by construction), for every $k\in [M]$, there exists a unique row-vector $\mathrm{C}_k$ of length $J$ such that ${\mathrm{U}_k = \mathrm{C}_k \mathrm{G}}$. 
%Let $k^{*}\in [M]$ be such that $\mathrm{W}_{k^{*}} = \mathrm{W}$. 
By (ii) and (iii), it is easy to verify that $\mathrm{U}_{k^{*}}\mathrm{X} = \mathrm{C}_{k^{*}}\mathrm{G}\mathrm{X}$ is equal to the user's demand $ \mathrm{Z} = \mathrm{V}\mathrm{X}_{\mathrm{W}}$, where $\mathrm{X} = [\mathrm{X}_1^{\transpose},\dots,\mathrm{X}_K^{\transpose}]^{\transpose}$ and $\mathrm{X}_{\mathrm{W}}$ is the submatrix of $\mathrm{X}$ formed by the rows indexed by $\mathrm{W}$. 

Then, the user sends the matrix $\mathrm{G}$ and the row-vectors $\{\mathrm{C}_k\}_{k\in [M]}$ to each of the servers.\vspace{0.125cm}   

\textbf{Step 2:} Upon receiving $\mathrm{G}$ and $\{\mathrm{C}_k\}_{k\in [M]}$, each server constructs $J$ ``coded messages'' $\mathrm{Y}_1,\dots,\mathrm{Y}_J$, where $\mathrm{Y}_i\triangleq \mathrm{G}_i\mathrm{X}$ for all $i\in [J]$, and $\mathrm{G}_i$ denotes the $i$th row of $\mathrm{G}$. 
Note that $\mathrm{Y}_i$'s are row-vectors of length $T$.
Let $\mathrm{Y} \triangleq [\mathrm{Y}_1^{\transpose},\dots,\mathrm{Y}_J^{\transpose}]^{\transpose}$.
Note that ${\mathrm{Y} = \mathrm{G}\mathrm{X}}$.
Each server then constructs $M$ ``coded combinations'' $\mathrm{Z}_1,\dots,\mathrm{Z}_M$, where ${\mathrm{Z}_k\triangleq \mathrm{C}_k \mathrm{Y}}$. 
Note that $\mathrm{Z}_k$'s are row-vectors of length $T$.
Since ${\mathrm{Z}_k = \mathrm{C}_k\mathrm{Y} = \mathrm{C}_k\mathrm{G}\mathrm{X} = \mathrm{U}_k\mathrm{X}}$, and the support of $\mathrm{U}_k$ is $\mathrm{W}_k$, the support of $\mathrm{Z}_k$ is $\mathrm{W}_k$.
Thus, $\mathrm{Z}_1,\dots,\mathrm{Z}_M$ are $M$ linear combinations of  $\mathrm{X}_1,\dots,\mathrm{X}_K$ with distinct supports $\mathrm{W}_1,\dots,\mathrm{W}_M$, respectively. 
Note, also, that $\mathrm{Z}_{k^{*}} = \mathrm{U}_{k^{*}}\mathrm{X} = \mathrm{V}\mathrm{X}_{\mathrm{W}}=\mathrm{Z}$, as discussed earlier.\vspace{0.125cm}    

\textbf{Step 3:}
Note that $\mathrm{Y}_1,\dots,\mathrm{Y}_J$ are linearly independent combinations of $X_1,\dots,X_K$, and $\mathbf{X}_1,\dots,\mathbf{X}_K$ are independent and uniformly distributed over $\mathbb{F}_q^{T}$. 
Thus, $\mathbf{Y}_1,\dots,\mathbf{Y}_J$ are independent and uniformly distributed over $\mathbb{F}_q^{T}$. 
Note, also, that $\mathrm{Z}_1,\dots,\mathrm{Z}_M$ are linear combinations of the coded messages $\mathrm{Y}_1,\dots,\mathrm{Y}_{J}$, with the coefficient vectors $\mathrm{C}_1,\dots,\mathrm{C}_M$, respectively. 
It is easy to see that $M\geq J$. 
Then, the user and the servers follow the PLC scheme of~\cite{SJ2018} with the coded messages $\mathrm{Y}_1,\dots,\mathrm{Y}_J$ as the ``independent messages'' and the coded combinations $\mathrm{Z}_1,\dots,\mathrm{Z}_M$ as the ``candidate linear combinations''. 
In order to recover the coded combination $Z_{k^{*}}$, for each $n\in [N]$ the user generates a query and sends it to server $n$, and server $n$ then sends back the corresponding answer to the user. 
Note that the scheme of~\cite{SJ2018} is applicable if the message length $T$, which is also the length of each coded message $\mathrm{Y}_i$ and the length of each coded combination $\mathrm{Z}_k$, is an integer multiple of $N^M$.  
The details of the construction of the user's queries and the servers' corresponding answers can be found in~\cite{SJ2018}. 

\begin{lemma}\label{lem:JPLC-Ach}
The Multi-Server Specialized GRS Code protocol is a capacity-achieving JPLC protocol.
\end{lemma}

\begin{proof}
The rate of the PLC scheme of~\cite{SJ2018} for the setting with $N$ servers, $J$ independent messages, and $M$ ($\geq J$) candidate linear combinations is given by ${(1+1/N+1/N^2+\dots+1/N^{J-1})^{-1}}$. 
In our case, ${J=K-D+1}$ and $M\geq J$. 
Thus, the rate of our protocol is given by ${(1+1/N+1/N^2+\dots+1/N^{K-D})^{-1}}$, which matches the converse bound (see Lemma~\ref{lem:JPLC-Conv}).

The recoverability of the user's demand $\mathrm{Z}$ is guaranteed because the PLC scheme of~\cite{SJ2018} ensures that the user can recover the coded combination $\mathrm{Z}_{k^{*}}$, which is equal to $\mathrm{Z}$ as discussed earlier. 
The proof of joint privacy is as follows. 
Exposing $\mathrm{Z}_1,\dots,\mathrm{Z}_M$ to the servers leaks no information about the support of the user's demand to any of the servers. 
This is because the single-server JPLT protocol of~\cite{EHS2021JointISIT} is guaranteed to protect the privacy of the support of the user's demand in the single-server setting, notwithstanding that $\mathrm{Z}_1,\dots,\mathrm{Z}_M$ are revealed to the server~\cite{EHS2021JointISIT}.
Given $\mathrm{Z}_1,\dots,\mathrm{Z}_M$, the privacy of the support of the user's demand is protected because the PLC scheme of~\cite{SJ2018} guarantees that the privacy of the index of the coded combination $\mathrm{Z}_{k^{*}}$ is protected.
\end{proof}

\section{Proof of Theorem~\ref{thm:MSIPLC}}\label{sec:MSIPLC}

\subsection{Converse Proof}\label{subsec:MSIPLC-Conv}
In this section, we establish an upper bound on the rate of IPLC protocols in terms of the parameters $N,K,D$, which proves the converse for Theorem~\ref{thm:MSIPLC}. 
The upper bound holds for any field size $q$ and any message length $T$.

\begin{lemma}\label{lem:IPLC-Conv}
The rate of any IPLC protocol for the setting with parameters $N,K,D$ is upper bounded by~\eqref{eq:IPLCCap}. 
%\begin{equation*}
%\left(1+\frac{1}{N}+\frac{1}{N^2}+\dots+\frac{1}{N^{\lceil K/D\rceil-1}}\right)^{-1}.    
%\end{equation*}
\end{lemma}

\begin{proof}
Similar to the proof of Lemma~\ref{lem:JPLC-Conv}, the main idea is to show a reduction. 
In this case, we show a reduction from the PIR-SI problem~\cite{KGHERS2020} to the IPLC problem. 
The PIR-SI problem with parameters $N,K,M$ is the same as the PIR-PSI problem defined in the proof of Lemma~\ref{lem:JPLC-Conv}, except that in this case it is only required to protect the privacy of the index of the user's uncoded demand, $i^{*}$, and the privacy of the support of the side information, $\mathrm{S}$, does not need to be protected.  
Relying on~\cite[Theorem~1]{LG2020CISS}, the capacity of PIR-SI---defined as the maximum achievable download rate---is given by 
\begin{equation}\label{eq:2}
\left(1+1/N+1/N^2+\dots+1/N^{\lceil K/(M+1)\rceil-1}\right)^{-1},    
\end{equation}
under the same assumptions on the distributions of $\mathbf{S}$ and $\boldsymbol{i}^{*}$ as those in the proof of Lemma~\ref{lem:JPLC-Conv} for the PIR-PSI problem. 
To show a reduction from PIR-SI to IPLC, we need to prove that the PIR-SI problem with parameters $N,K,M$ can be solved by any IPLC protocol for the setting with $N$ servers, $K$ messages, and demand's support size $D=M+1$. 
Using this reduction and relying on~\eqref{eq:2}, a simple proof by contradiction similar to that in the proof of Lemma~\ref{lem:JPLC-Conv} yields the result of the lemma. 
%Suppose that the rate of an IPLC protocol for the setting with parameters $N,K,D$ defined as above is higher than ${1/(1+1/N+\dots+1/N^{\lceil K/D\rceil -1})}$, or equivalently, ${1/(1+1/N+\dots+1/N^{\lceil K/(M+1)\rceil-1})}$. 
%Solving the PIR-SI problem via this IPLC protocol, one can then achieve a higher rate than the capacity of PIR-SI, which is a contradiction.
%Below we show the reduction. %, we proceed as follows. 

To show a reduction, we follow the exact same line as in the proof of Lemma~\ref{lem:JPLC-Conv}, except that in this case we show how to solve the PIR-SI problem using an IPLC protocol, instead of solving the PIR-PSI problem using a JPLC protocol. 
The proof of recoverability of the user's uncoded  demand is the same as before. 
To prove that the privacy of the index of the uncoded demand is protected, we need to show that ${\Pr(\boldsymbol{i}^{*}=i|\mathbf{Q}_n = \mathrm{Q}_n) = \Pr(\boldsymbol{i}^{*}=i) = 1/K}$ for all ${i\in [K]}$.

For each $i\in [K]$, we denote by $\mathbbmss{W}_i$ the set of all $\tilde{\mathrm{W}}\in \mathbbmss{W}$ such that $i\in \tilde{\mathrm{W}}$.  
For all $i\in [K]$, for all $n\in [N]$, we have 
\begin{align}
& \Pr(\boldsymbol{i}^{*}=i|\mathbf{Q}_n = \mathrm{Q}_n) \nonumber \\
& = \Pr(\boldsymbol{i}^{*}=i, i\in \mathbf{W}|\mathbf{Q}_n = \mathrm{Q}_n) \label{eq:3}\\
& = \Pr(i\in \mathbf{W}|\mathbf{Q}_n= \mathrm{Q}_n)\Pr(\boldsymbol{i}^{*}=i|\mathbf{Q}_n = \mathrm{Q}_n,i\in \mathbf{W}) \label{eq:4}\\
& = \frac{D}{K}\times \Pr(\boldsymbol{i}^{*}=i|\mathbf{Q}_n = \mathrm{Q}_n,i\in \mathbf{W}) \label{eq:5}\\
& = \frac{D}{K}\times \sum_{\tilde{\mathrm{W}}\in \mathbbmss{W}} \Pr(\boldsymbol{i}^{*}=i, \mathbf{W} = \tilde{\mathrm{W}}|\mathbf{Q}_n = \mathrm{Q}_n,i\in \mathbf{W}) \label{eq:6}\\
& = \frac{D}{K}\times \sum_{\tilde{\mathrm{W}}\in \mathbbmss{W}_i} \Pr(\boldsymbol{i}^{*}=i, \mathbf{W} = \tilde{\mathrm{W}}|\mathbf{Q}_n = \mathrm{Q}_n,i\in \mathbf{W}) \label{eq:7}\\
& = \frac{D}{K}\times \sum_{\tilde{\mathrm{W}}\in \mathbbmss{W}_i} \Pr(\mathbf{W} = \tilde{\mathrm{W}}|\mathbf{Q}_n = \mathrm{Q}_n,i\in \mathbf{W}) \nonumber \\
& \quad \quad \quad \quad \quad \quad \quad \times \Pr(\boldsymbol{i}^{*}=i|\mathbf{Q}_n = \mathrm{Q}_n,\mathbf{W} = \tilde{\mathrm{W}}) \label{eq:8}\\
& = \frac{D}{K}\times \sum_{\tilde{\mathrm{W}}\in \mathbbmss{W}_i} \Pr(\mathbf{W} = \tilde{\mathrm{W}}|\mathbf{Q}_n = \mathrm{Q}_n,i\in \mathbf{W}) \nonumber \\
& \quad \quad \quad \quad \quad \quad \quad \times \Pr(\boldsymbol{i}^{*}=i|\mathbf{W} = \tilde{\mathrm{W}}) \label{eq:9}\\
& = \frac{D}{K}\times \frac{1}{D} \times \sum_{\tilde{\mathrm{W}}\in \mathbbmss{W}_i} \Pr(\mathbf{W} = \tilde{\mathrm{W}}|\mathbf{Q}_n = \mathrm{Q}_n,i\in \mathbf{W}) \label{eq:10}\\
& = \frac{1}{K}, \label{eq:11}
\end{align} 
where~\eqref{eq:3} holds because $\mathbf{W} = \{\boldsymbol{i}^{*}\}\cup \mathbf{S}$;
~\eqref{eq:4} follows from the chain rule of probability;
~\eqref{eq:5} holds because any IPLC protocol satisfies the individual privacy condition, i.e., ${\Pr(i\in \mathbf{W}|\mathbf{Q}_n = \mathrm{Q}_n) = \Pr(i\in \mathbf{W})}$, and ${\Pr(i\in \mathbf{W})=D/K}$;
~\eqref{eq:6} follows from the law of total probability;
~\eqref{eq:7} holds because   ${\Pr(\mathbf{W}=\tilde{\mathrm{W}}|i\in\mathbf{W})=0}$ for all ${\tilde{\mathrm{W}}\in \mathbbmss{W}_i}$;
~\eqref{eq:8} follows from the chain rule of probability;
~\eqref{eq:9} holds because by the same arguments as in the proof of Lemma~\ref{lem:JPLC-Conv}, for all ${\tilde{\mathrm{W}}\in \mathbbmss{W}_i}$, ${\mathbf{Q}_n=\mathrm{Q}_n}$ and ${\boldsymbol{i}^{*}=i}$ are conditionally independent given ${\mathbf{W} = \tilde{\mathrm{W}}}$;
~\eqref{eq:10} holds because ${\Pr(\boldsymbol{i}^{*}=i|\mathbf{W}=\tilde{\mathrm{W}}) = 1/D}$ for all ${\tilde{\mathrm{W}}\in \mathbbmss{W}_i}$, as shown in the proof of Lemma~\ref{lem:JPLC-Conv}; 
and~\eqref{eq:11} holds because given ${\mathbf{Q}_n = \mathrm{Q}_n}$ and ${i\in \mathbf{W}}$, it follows that ${\mathbf{W}\in \mathbbmss{W}_i}$, i.e., 
${\Pr(\mathbf{W}\in \mathbbmss{W}_i|\mathbf{Q}_n = \mathrm{Q}_n,i\in \mathbf{W})} = 1$,
and 
hence, ${\sum_{\tilde{\mathrm{W}}\in\mathbbmss{W}_i} \Pr(\mathbf{W} = \tilde{\mathrm{W}}|\mathbf{Q}_n = \mathrm{Q}_n,i\in \mathbf{W})=1}$ because ${\mathbf{W}=\tilde{\mathrm{W}}}$ for all ${\tilde{\mathrm{W}}\in \mathbbmss{W}_i}$ are disjoint events.

By~\eqref{eq:11}, all $i\in [K]$ are equally likely to be the index of the user's uncoded demand, from each server's perspective. 
Thus, the privacy of index of the uncoded demand is protected.
%This completes the proof.
\end{proof}

\subsection{Achievability Scheme}\label{subsec:MSIPLC-Ach}
In this section, we present an IPLC protocol, termed \emph{Multi-Server Partition-and-Code with Partial Interference Alignment}, for all parameters $N,K,D$ such that $R=0$ or $R\mid D$, where $R \triangleq K\pmod D$. 
When $R=0$, our protocol is capacity-achieving for any $q\geq 2$ and any $T$ that is an integer multiple of $N^{M_1}$, where $M_1\triangleq K/D$. 
When $R\mid D$, our protocol achieves the capacity for any $q\geq D/R+1$ and any $T$ that is an integer multiple of $N^{M_2}$, where $M_2\triangleq \lfloor K/D\rfloor+D/R$.
An example of this protocol is given in Section~\ref{sec:EX}. 

The %Multi-Server Partition-and-Code with Partial Interference Alignment 
protocol consists of three steps. 
Steps 2 and 3 of this protocol are the same as those in our JPLC protocol (and hence omitted to avoid repetition), expect that in this protocol, the number of ``coded messages'' is $\lceil K/D\rceil$, and the number of ``coded combinations'' is $M_1$ or $M_2$ when ${R=0}$ or ${R\mid D}$, respectively.  
Step 1 of this protocol, however, differs from Step 1 of our JPLC protocol, and is as described below.\vspace{0.125cm} 

\textbf{Step 1:} 
Utilizing the single-server IPLT protocol of~\cite{EHS2021IndividualISIT} for the special case in which one linear combination is required, the user first constructs a specific $J\times K$ matrix $\mathrm{G}$, where $J\triangleq \lceil K/D \rceil$. 
The steps of this protocol are omitted for brevity, and only the matrix $\mathrm{G}$ being constructed is presented below.     

Recall that $\mathrm{W}$ and $\mathrm{V}$ denote the support and the coefficient vector of the user's demand $\mathrm{Z}$. 
Suppose ${\mathrm{W}=\{i_1,\dots,i_D\}}$, ${[K]\setminus \mathrm{W} = \{i_{D+1},\dots,i_K\}}$, and ${\mathrm{V} = [v_1,\dots,v_D]}$.

\begin{figure*}
\begin{equation}\label{eq:14}
\setlength\arraycolsep{2.15pt}
\begin{bmatrix}
\bovermat{$D$}{\alpha_{1,1} & \cdots & \alpha_{1,D}}&  & &  & & &  &  &  & \\
 & &  & \bovermat{$D$}{\alpha_{2,1} & \cdots & \alpha_{2,D}} & & &  &  &  & \\
 &  & &  &  &  & & \ddots & & & & \\
 &  &  &  &  &  & & & & \bovermat{$D$}{\alpha_{\frac{K}{D},1} & \cdots & \alpha_{\frac{K}{D},D}}  
\end{bmatrix}%\vspace{-0.25cm}
\end{equation}
\end{figure*}

%\vspace{2.1cm}
\begin{figure*}
\begin{equation}\label{eq:15}
\setlength\arraycolsep{3pt}
\begin{bmatrix}
\bovermat{$D$}{\alpha_{1,1}&\cdots&\alpha_{1,D}}& & & & & & & & & & & & & & &\\
& & & &\ddots & & & & & & & & & & & & &\\
&  &  & & & & \bovermat{$D$}{\alpha_{n,1}&\cdots&\alpha_{n,D}} & & & & & & & & &\\
%&&&&&&&&&&&&&&&&&\\
&  &  & & & & & & & \alpha_{n+1,1} &\cdots&\alpha_{n+1,R}& %\bovermat{$R$}{\beta_{2,1}&\cdots&\beta_{2,R}}
&&\cdots&&\alpha_{n+m,1}&\cdots&\alpha_{n+m,R} \\
&  &  & & & & & & &
\bundermat{$R$}{\alpha_{n+1,1}\omega_1&\cdots&\alpha_{n+1,R}\omega_1} & %\beta_{2,1}\omega_2&\cdots&\beta_{2,R}\omega_2
&&\cdots&& \bundermat{$R$}{\alpha_{n+m,1}\omega_{m}&\cdots&\alpha_{n+m,R}\omega_{m}}
\end{bmatrix}
\end{equation}
\end{figure*}

\vspace{0.25cm}
\emph{Case of $R=0$:} In this case, the matrix $\mathrm{G}$ is obtained by applying a carefully chosen permutation $\pi$---specified below, on the columns of the matrix given by~\eqref{eq:14} with parameters $\alpha_{i,j}$'s---defined shortly. 
Note that $J=K/D$.

%, where $\pi$ and $\alpha_{i,j}$'s are defined as follows. 

For a randomly chosen permutation $\sigma$ on $[D]$ and a randomly chosen ${i^{*}\in [K/D]}$, 
\begin{itemize}
    \item ${\alpha_{i^{*},j} = v_{\sigma(j)}}$ for all ${j\in [D]}$.
    \item ${\alpha_{i,j}}$'s for all ${i\in [K/D]\setminus \{i^{*}\}}$ and all ${j\in [D]}$ are randomly chosen (with replacement) elements from ${\mathbb{F}_q^{\times}}$.
    \item $\pi$ is a randomly chosen permutation on $[K]$ such that ${\pi((i^{*}-1)D+j) = i_{\sigma(j)}}$ for all ${j\in [D]}$.
\end{itemize}

% = \{\pi(j)\}_{j\in [(k-1)D+1:kD]}
% {\{\pi((k-1)D+1),\dots,\pi(kD)\}
%Let $M_1\triangleq K/D$, and let $\mathrm{W}_k\triangleq \{\pi(j)\}_{j\in [(k-1)D+1:kD]}$ for all $k\in [M_1]$.

Let $\mathrm{W}_1,\dots,\mathrm{W}_{M}$ be an arbitrary ordering of the elements in $\mathbbmss{W}$. 
Without loss of generality, assume that $\mathrm{W}_k\triangleq \{\pi(j)\}_{j\in [(k-1)D+1:kD]}$ for all $k\in [M_1]$, where $M_1\triangleq K/D$. 
Note that $M_1 = J$.

As shown in~\cite{EHS2021IndividualISIT}, the matrix $\mathrm{G}$ has the following properties: 
\begin{itemize}
    \item[(i)] For every $k\in [M_1]$, the row space of $\mathrm{G}$ contains a unique row-vector $\mathrm{U}_k$ with support $\mathrm{W}_k$ and first nonzero coordinate equal to the first coordinate of $\mathrm{V}$ (i.e., $v_1$), and for every $k\in [M_1+1:M]$, the row space of $\mathrm{G}$ does not contain any row-vector with support $\mathrm{W}_k$. \item[(ii)] There exists a unique $k^{*}\in [M_1]$ such that $\mathrm{W}_{k^{*}} = \mathrm{W}$.
    \item[(iii)] The row-vector $\mathrm{U}_{k^{*}}$ with support $\mathrm{W}_{k^{*}} = \mathrm{W}$, when restricted to its nonzero coordinates, is equal to $\mathrm{V}$.
\end{itemize}

Using (i)-(iii) and similar arguments as in Step~1 of our JPLC protocol, it follows that, for every ${k\in [M_1]}$, 
%By the property (i), for every $k\in [M_1]$, the row-vector $\mathrm{U}_k$ can be written as a unique linear combination of the rows of $\mathrm{G}$, i.e., $\mathrm{G}_1,\dots,\mathrm{G}_{J}$. 
%Since $\mathrm{G}_1,\dots,\mathrm{G}_{J}$ are linearly independent (by construction), 
there exists a unique row-vector $\mathrm{C}_k$ of length $J$ such that $\mathrm{U}_k = \mathrm{C}_k \mathrm{G}$, and  
%Let $k^{*}\in [M]$ be such that $\mathrm{W}_{k^{*}} = \mathrm{W}$. 
%By the properties (ii) and (iii), it is easy to verify that 
$\mathrm{U}_{k^{*}}\mathrm{X} = \mathrm{C}_{k^{*}}\mathrm{G}\mathrm{X}$ is equal to the user's demand $ \mathrm{Z} = \mathrm{V}\mathrm{X}_{\mathrm{W}}$. % where $\mathrm{X} = [\mathrm{X}_1^{\transpose},\dots,\mathrm{X}_K^{\transpose}]^{\transpose}$ and $\mathrm{X}_{\mathrm{W}}$ is the submatrix of $\mathrm{X}$ formed by the rows indexed by $\mathrm{W}$. 

Then, the user sends the matrix $\mathrm{G}$ and the row-vectors $\{\mathrm{C}_k\}_{k\in [M_1]}$ to each of the servers.

\vspace{0.25cm}
\emph{Case of $R\mid D$:} In this case, the matrix $\mathrm{G}$ is obtained by applying a carefully designed permutation $\pi$---specified below, on the columns of the matrix given by~\eqref{eq:15} with parameters $\alpha_{i,j}$'s and $\omega_i$'s---defined shortly. 
Note that ${J=(K-R)/D+1}$. 

For simplifying the notation, let ${n\triangleq (K-R)/D-1}$ and 
${m\triangleq D/R+1}$.
The parameters $\omega_1,\dots,\omega_{m}$ are $m$ arbitrary distinct elements from $\mathbb{F}_q$. 
The parameters $\alpha_{i,j}$'s and the permutation $\pi$ are determined using one of two algorithms, referred to as Algorithms 1 and 2, where Algorithm 1 or 2 is selected with probability $D/K$ or $1-D/K$, respectively. 

\emph{Algorithm 1:} For a randomly chosen permutation $\sigma$ on $[D]$ and a randomly chosen $i^{*}\in [n]$, 
\begin{itemize}
    \item $\alpha_{i^{*},j}=v_{i_{\sigma(j)}}$ for all $j\in [D]$. 
    \item $\alpha_{i,j}$'s for all $i\in [n]\setminus \{i^{*}\}$ and all $j\in [D]$ are chosen randomly (with replacement) from $\mathbb{F}_q^{\times}$.
    \item $\alpha_{n+i,j}$'s for all $i\in [m]$ and all $j\in [R]$ are chosen randomly (with replacement) from $\mathbb{F}_q^{\times}$. 
    \item $\pi$ is a randomly chosen permutation on $[K]$ such that ${\pi((i^{*}-1)D+j) = i_{\sigma(j)}}$ for all $j\in [D]$. 
\end{itemize}

\emph{Algorithm 2:} For a randomly chosen permutation $\sigma$ on $[D]$ and a randomly chosen ${i_{*}\in [m]}$, 
\begin{itemize}
    \item $\alpha_{i,j}$'s for all $i\in [n]$ and all $j\in [D]$ are chosen randomly (with replacement) from $\mathbb{F}_q^{\times}$.
    \item $\alpha_{n+i,j} = v_{\sigma((i-1)R+j)}/(\omega_{i_{*}}-\omega_i)$ for all ${i\in [i_{*}-1]}$ and all ${j\in [R]}$.
    \item $\alpha_{n+i_{*},j}$ for all $j\in [R]$ are  chosen randomly (with replacement) from $\mathbb{F}_q^{\times}$.
    \item $\alpha_{n+i,j} = v_{\sigma((i-1)R-R+j)}/(\omega_{i_{*}}-\omega_i)$ for all ${i\in [i_{*}+1:m]}$ and all ${j\in [R]}$.
    
    \item $\pi$ is a randomly chosen permutation on $[K]$ such that ${\pi(nD+j) = i_{\sigma(j)}}$ for all ${j\in [(i_{*}-1)R]}$, and ${\pi(nD+R+j) = i_{\sigma(j)}}$ for all ${j\in [(i_{*}-1)R+1:D]}$. 
\end{itemize}

Let $\mathrm{W}_1,\dots,\mathrm{W}_{M}$ be an arbitrary ordering of the elements in $\mathbbmss{W}$. 
Assume, without loss of generality, that $\mathrm{W}_k\triangleq \{\pi(j)\}_{j\in [(k-1)D+1:kD]}$ for all $k\in [n]$, and 
$\mathrm{W}_k\triangleq \{\pi(j)\}_{j\in [nD+(k-n-1)R]} \cup \{\pi(j)\}_{j\in [nD+(k-n)R+1:K]}$ for all $k\in [n+1:n+m]$. 
Let $M_2\triangleq n+m$. 
Note that $M_2 = (K-R)/D+D/R > (K-R)/D+1 = J$. %because $D/R>1$. 

As shown in~\cite{EHS2021IndividualISIT}, the matrix $\mathrm{G}$ satisfies the same set of properties as those for the case of $R=0$, when $M_1$ is replaced by $M_2$. 
This implies that, for every $k\in [M_2]$, there is a unique $\mathrm{C}_k$ such that $\mathrm{U}_k = \mathrm{C}_k \mathrm{G}$, and $\mathrm{U}_{k^{*}}\mathrm{X} = \mathrm{C}_{k^{*}}\mathrm{G}\mathrm{X}= \mathrm{V}\mathrm{X}_{\mathrm{W}}=\mathrm{Z}$.

The user then sends the matrix $\mathrm{G}$ and the row-vectors $\{\mathrm{C}_k\}_{k\in [M_2]}$ to each of the servers.

%Note that $\mathrm{G}$ is a $J\times K$ matrix which generates a $[K,J]$ GRS code, and $\{\alpha_{j}\}_{j\in [K]}$ and $\{\omega_{j}\}_{j\in [K]}$ are the multipliers and the evaluation points of this GRS code, respectively. 

% (up to scalar multiplication) 

\begin{lemma}\label{lem:IPLC-Ach}
The Multi-Server Partition-and-Code with Partial Interference Alignment protocol is a capacity-achieving IPLC protocol.
\end{lemma}

\begin{proof}
The proof of optimality of the rate follows from the result of Lemma~\ref{lem:IPLC-Conv} and the exact same arguments as those in the proof of Lemma~\ref{lem:JPLC-Ach}, except that in this case $J=\lceil K/D \rceil$, %(instead of $K-D+1$), 
and $M$ is replaced by $M_1$ or $M_2$ for the case of $R=0$ or $R\mid D$, respectively. 
%Similar to the proof of Lemma~\ref{lem:JPLC-Ach}, 
The user's demand can be recovered because the PLC scheme of~\cite{SJ2018} ensures the recoverability of the coded combination required by the user. 
The proof of individual privacy is as follows. 
To avoid repetition, we only present the proof for the case of $R=0$. 
The proof for the case of $R\mid D$ follows from the exact same line, when $M_1$ is replaced by $M_2$.  
Let $\mathrm{Z}_1,\dots,\mathrm{Z}_{M_1}$ be the ``coded combinations'' constructed in Step 2 of the protocol. 
%Notwithstanding that  $\mathrm{Z}_1,\dots,\mathrm{Z}_M$ are revealed to each server, the (individual) privacy of every index in the support of the user's demand is protected. 
Revealing $\mathrm{Z}_1,\dots,\mathrm{Z}_{M_1}$ to each of the servers does not violate the individual privacy condition. 
This is because the single-server IPLT protocol of~\cite{EHS2021IndividualISIT}---used for constructing $\mathrm{Z}_1,\dots,\mathrm{Z}_{M_1}$, exposes these coded combinations to the server, and is guaranteed to protect the privacy of every index in the support of the user's demand. 
Given $\mathrm{Z}_1,\dots,\mathrm{Z}_{M_1}$, the individual privacy condition is satisfied, because the PLC scheme of~\cite{SJ2018} protects the privacy of the index of the coded combination required by the user.
\end{proof}

\section{examples}\label{sec:EX}
In this section, we provide an illustrative example of each of the proposed protocols. 

\begin{example}\label{ex-1}
\normalfont 
%In the following, an example of the MS-JPLC protocol is presented. 
Consider a JPLC setting in which there are $N = 2$ servers each storing ${K=3}$ messages $X_1,X_2,X_3\in\mathbb{F}^{8}_{3}$, and 
the user wants to compute one linear combination of ${D=2}$ messages $X_1$ and $X_3$, say, $Z = X_1+2X_3$. 
Note that for this example, $\mathrm{W}=\{1,3\}$ and $\mathrm{V}=[1,2]$. 
Using the notation in Section~\ref{subsec:MSJPLC-Ach},  $(i_1,i_2,i_3) = (1,3,2)$, $(v_1,v_2)=(1,2)$, and 
\[\pi = \begin{pmatrix} 1 & 2 & 3\\ 1 & 3 & 2\end{pmatrix}.\] 
Taking $v_3=1$ and $(\omega_1,\omega_2,\omega_3)=(0,1,2)$ in Step 1 of the proposed JPLC protocol, the user constructs the matrix 
%First, the user constructs a $[3,2]$ MDS code using the single-server JPLT-I protocol in~~\cite{HES2021JointJournal}, and obtains 
\begin{equation*}
\mathrm{G}=
\begin{bmatrix}
1 & 2 & 1\\
0 & 1 & 1
\end{bmatrix}.
\end{equation*}
%Note that this matrix forms the query of the user constructed by the \emph{Specialized MDS Code protocol} in~\cite{HES2021JointJournal}, for the same set of parameters as in this example.
%where the user wishes to privately compute one linear combination ${Z = X_1+2X_3}$.
%To avoid repetition, we skipped the steps of the achievability scheme in~\cite{HES2021JointJournal}, and only presented the generated query, denoted by $\tilde{Q}$.
%, using that scheme.
%%%%%%
It is easy to verify that $\mathrm{G}$ generates a $[3,2]$ maximum distance separable (MDS) code.
%Since $\mathrm{G}$ is an MDS matrix, it generates a [3,2] MDS code, and hence the rows of $\mathrm{G}$ form a basis for the code space generated by $\mathrm{G}$. 
Since the minimum distance of this code is $2$, for every $2$-subset of $\{1,2,3\}$, i.e., $\mathrm{W}_1 = \{1,2\}$, $\mathrm{W}_2 = \{1,3\}$, and $\mathrm{W}_3 = \{2,3\}$, the row space of $\mathrm{G}$ contains a unique vector (up to scalar multiplication) with support $\mathrm{W}_1$, $\mathrm{W}_2$, and $\mathrm{W}_3$, respectively, 
e.g., the vectors $\mathrm{U}_1=[1,1,0]$, 
$\mathrm{U}_2=[1,0,2]$, and
$\mathrm{U}_3=[0,1,1]$, respectively. 
Let $\mathrm{X} = [X^{\transpose}_1,X^{\transpose}_2,X^{\transpose}_3]^{\transpose}$. 
Let $\mathrm{Y} = \mathrm{G}\mathrm{X}$, and let $\mathrm{Y}_1,\mathrm{Y}_2$ denote the rows of $\mathrm{Y}$, 
\begin{align*}
\mathrm{Y}_{1}& =X_1+2X_2+X_3,\\
\mathrm{Y}_{2}&=X_2+X_3.
\end{align*}
Note that $\mathrm{Y}_1$ and $\mathrm{Y}_2$ are linearly independent combinations of the messages $X_1,X_2,X_3$.
Thus, the coded combinations $\mathrm{Z}_1, \mathrm{Z}_2,  \mathrm{Z}_3$ given by
\begin{align*}
\mathrm{Z}_1 & = \mathrm{U}_1\mathrm{X} = X_1+X_2, \\    
\mathrm{Z}_2 & = \mathrm{U}_2\mathrm{X} = X_1+2X_3, \\
\mathrm{Z}_3 & = \mathrm{U}_3\mathrm{X} = X_2+X_3,    
\end{align*} can be written in terms of the coded messages $\mathrm{Y}_1,\mathrm{Y}_2$ as 
\begin{align*}
\mathrm{Z}_1 & = \mathrm{Y}_1+2\mathrm{Y}_2,\\    
\mathrm{Z}_2 & = \mathrm{Y}_1+\mathrm{Y}_2,\\
\mathrm{Z}_3 & = \mathrm{Y}_2.
\end{align*}
Let $\mathrm{C}_1 = [1,2]$, $\mathrm{C}_2 = [1,1]$, and $\mathrm{C}_3 = [0,1]$ be the coefficient vectors corresponding to $\mathrm{Z}_1$, $\mathrm{Z}_2$, and $\mathrm{Z}_3$, respectively, i.e., $\mathrm{Z}_k = \mathrm{C}_k \mathrm{Y}$ for $k\in [3]$. 
Note that the user's demand $\mathrm{Z} =X_1+2X_3= \mathrm{Z}_2$. 

Next, the user sends the matrix $\mathrm{G}$ and the vectors $\mathrm{C}_1,\mathrm{C}_2,\mathrm{C}_3$ to each of the servers. 
%%%%%%
%This implies that $Y_{1}=X_1+2X_2+X_3$ and $Y_{2}=X_2+X_3$, obtained by multiplying the first and the second rows of matrix $\mathrm{G}$ by the message matrix $\mathrm{X}$, are linearly independent from each other, and all $M=C_{3,2}=2$ linear combinations (up to scalar multiplication) whose support is of length $D=2$ can be obtained by linearly combining $Y_1$ and $Y_2$.
%$Y_{1}=2X_1+X_2+ 2X_3$ and $Y_{2}=X_1+2X_3$, obtained by multiplying the first and the second rows of matrix $\tilde{Q}$ by the message matrix $\mathrm{X}$, form a basis for the function space containing all linear combinations (up to scalar multiplication) whose support is of length two. 
%More specifically, for this example, the three unique linear combinations denoted by $Z_1$,$Z_2$, and $Z_3$ with support of length $2$, i.e., $\{1,2\}$,$\{1,3\}$, and $\{2,3\}$, respectively, are
%\begin{align*}
%Z_1 &=Y_{1}+2Y_{2}=X_1
%+X_2,\\ 
%Z_2 &=Y_{1}+Y_{2}=X_1
%+2X_3,\\
%Z_3 &=Y_{2}=X_2+X_3.
%\end{align*}
%\begin{comment}
%$Z_1=Y_{1}
%+2Y_{2}$, $Z_2=Y_{2}$, and $Z_3=2Y_{1}
%+2Y_{2}$
%\end{comment}
%For this example, there are three unique linear combinations up to scalar multiplication whose support is of length two. 
%Next, the user sends to both servers the combination coefficient vectors $[1,1,0]$,$[1,0,2]$, and $[0,2,2]$ which correspond to $Z_1$, $Z_2$, and $Z_3$, respectively. 
Note that the user provides the servers with this information so that the servers know the set of all coded combinations among which the user wishes to compute one. %Note that although this set of functions becomes known to the servers, but due to the fact that these linear combinations are generated by using the single-server JPLT protocol, from the perspective of each server, every $D=2$-subset of messages are equally likely to be the $2$ messages appearing in the user's demand. Hence, the privacy condition is not violated.  
%Note that $Y_1$ and $Y_2$ are linearly independent, and the collection of them can be viewed as a dataset stored at each server in the private computation problem in~\cite{SJ2018}. Then, $Z_1$, $Z_2$, and $Z_3$ can be viewed as $M=3$ linear combinations of $Y_1$ and $Y_2$. In this setting, 
Then, the user and the servers follow the PLC scheme of~\cite{SJ2018} for $N=2$ servers, $M=3$ coded combinations $\mathrm{Z}_1,\mathrm{Z}_2,\mathrm{Z}_3$, each of which is a linear combination of $J=2$ (independent) coded messages  $\mathrm{Y}_1,\mathrm{Y}_2$, so that the user can privately recover the coded combination $\mathrm{Z}_2$. 
Note that $\mathrm{Z}_1,\mathrm{Z}_2,\mathrm{Z}_3$ each contains $T=8$ symbols (from $\mathbb{F}_3$). 
Since $T$ is an integer multiple of $N^M=8$, the scheme of~\cite{SJ2018} is applicable to this setting. 
For all $k\in [3]$ and all $i\in [8]$, let $\mathrm{Z}_k(i)$ denote the $i$th symbol of $\mathrm{Z}_k$, and for all $i\in [8]$, let $(a_i,b_i,c_i)\triangleq (s_i\mathrm{Z}_1(\tau(i)),s_i\mathrm{Z}_2(\tau(i)),s_i\mathrm{Z}_3(\tau(i)))$ for a randomly chosen permutation $\tau$ on $[8]$ and a randomly chosen integer $s_i\in \{-1,+1\}$. 
The user's queries generated by the scheme of~\cite{SJ2018} for this example are presented in Table~\ref{tab:ex1}. 
The details are omitted to avoid repetition. 
%For instance, the user queries the symbol $a_1$ from server 1, and server 1 answers to this query by sending the symbol $a_1$ back to the user. 
%Note that each of $\mathrm{Z}_1,\mathrm{Z}_2,\mathrm{Z}_3$ contains $T=8$ symbols (from $\mathbb{F}_3$). 
%The user's queries to servers 1 and 2 are given in Table~\ref{tab:ex1}.  

%to privately compute one linear combination ($Z_2$ in this example) from the set of $M=3$ linear combinations $Z_1$, $Z_2$, and $Z_3$. 
%In this example, we avoid repeating the details and steps of the private computation achievability scheme.
%The user subpacketizes each linear combination $Z_1$, $Z_2$, and $Z_3$ into $T=N^M=2^3=8$ symbols. 
%Let $a$, $b$, and $c$ denote $Z_1$, $Z_2$, and $Z_3$, respectively, after subpacketization and performing \emph{index assignment} and \emph{sign assignment} as part of the private computation achievability scheme in~\cite{SJ2018}. We represent the symbols of $a,b$, and $c$ as follows:
%\begin{align*}
%    a&=(a_1,\dots,a_8), \quad b\hspace{-0.3cm}&=(b_1,\dots,b_8), \quad c &=(c_1,\dots,c_8).
%\end{align*}
%The query constructed by the private computation scheme is given in Table\ref{table:ex1}.

%%%%
\begin{table}
\caption{The user's queries to the servers for Example~1.}
\label{tab:ex1}
\centering
\resizebox{0.55\columnwidth}{!}{%
\begin{tabular}{| c | c |}
  \hline 
  \textbf{Server 1} & \textbf{Server 2} \\
  \hline
  $a_1, b_1, c_1$ & $a_2, b_2, c_2$ \\
  \hline
  $a_2-b_3$ & $a_1-b_5$ \\
  %\hline
  $b_4+c_2$ & $b_6+c_1$  \\
  %\hline
  $a_4-c_3$ & $a_6-c_5$  \\
  \hline
  $a_6-b_7-c_5$ & $a_4-b_8-c_3$  \\
  \hline 
\end{tabular}
}
\end{table}

Since any one of $\mathrm{Z}_1,\mathrm{Z}_2,\mathrm{Z}_3$ can be written as a linear combination of the other two, 
%Since $Z_3$ is a linear combination of $Z_1$ and $Z_2$, it is easy to see that symbols $c_1$ and $c_2$ in the queries from server 1 and 2 are redundant, respectively. 
any two of $a_1,b_1,c_1$ suffice to recover the other one. 
Similarly, any two of $a_2,b_2,c_2$ suffice to recover the other one. 
Thus, server~1 answers by sending 2 of the symbols $a_1,b_1,c_1$, say, $a_1$ and $b_1$, and the 4 remaining coded symbols $a_2-b_3,b_4+c_2,a_4-c_3,a_6-b_7-c_5$. 
Similarly, server~2's answer consists of 2 of the symbols $a_2,b_2,c_2$, say, $a_2$ and $b_2$, and the 4 remaining coded symbols $a_1-b_5,b_6+c_1,a_6-c_5,a_4-b_8-c_3$. 
Thus, the total number of symbols being downloaded from both servers is $12$, and 
the rate of this scheme is $8/12 = 2/3$, 
which matches the upper bound in Lemma~\ref{lem:JPLC-Conv} for $N=2$, $K=3$, and $D=2$, 
i.e.,  ${(1+1/N+\dots+1/N^{K-D})^{-1} = (1+1/2)^{-1} = 2/3}$.

%As explained in~\cite{SJ2018}, for this example, the number of effective downloaded symbols from each server is 6. In order to minimize the total number of bits being downloaded from the servers (without violating the privacy and achievability conditions), the user takes a randomized procedure to remove the redundant bits from the queries. %The user does this with a slight modification in the queries from the server as follows.
%More specifically, the user randomly selects two symbols from $a_1,b_1,c_1$, say, $a_1$ and $b_1$, and removes the dependent symbol $c_1$ from the query of server $1$.  Similarly, $c_2$ is removed from the query from server $2$. Then, the user sends to servers $1$ and $2$ their corresponding query. 
%Upon receiving the query, each server computes the coded symbols queried by the user, and sends them back as the answer to the user. 
To show that the coded combination $\mathrm{Z}_2$ can be recovered from the answer, it suffices to show that the symbols ${b_1,\dots,b_8}$ are recoverable from the answer.  
%the desired message symbols, i.e., $b_1,\dots,b_8$,
%the user proceeds as follows. 
From Table~\ref{tab:ex1}, it can be seen that the user can readily recover $b_1$ and $b_2$ from the answer.
%,from server $1$ and $2$, respectively. 
Since the answer also contains $a_1$ and $a_2$, the user can locally compute $c_1$ and $c_2$. 
Thus, the user can recover ${b_3,b_4,b_5,b_6}$ by subtracting off the contribution of ${a_2,c_2,a_1,c_1}$ from ${a_2-b_3,b_4+c_2,a_1-b_5,b_6+c_1}$, respectively. 
%, the user can solve for $b_3,b_4,b_5,b_6$ using the coded symbols $a_1-b_5,b_6+c_1,a_2-b_3,b_4+c_2$, respectively. 
Since the answer also contains $a_6-c_5$ and $a_4-c_3$, the user can recover $b_7$ and $b_8$ by subtracting off the contribution of $a_6-c_5$ and $a_4-c_3$ from $a_6-b_7-c_5$ and $a_4-b_8-c_3$, respectively. 
\end{example}

\begin{example}\label{ex-2}
\normalfont 
%This example illustrates the MS-IPLC protocol for the case that $R$ divides $D$. 
Consider an IPLC setting in which there are $N = 2$ servers each storing ${K=5}$ messages $X_1,\dots,X_{5}\in\mathbb{F}^{16}_{3}$, and the user wants to compute one linear combination of ${D=2}$ messages $X_1$ and $X_3$, say, $Z = X_1+2X_3$.
Note that $R = K \pmod D = 1$, and hence, $R\mid D$. 
For this example, $\mathrm{W}=\{1,3\}$ and $\mathrm{V}=[1,2]$. 
%\begin{equation*}
%$\mathrm{V} = 
%\begin{bmatrix}
%1 & 2
%\end{bmatrix}.$
%\end{equation*}
Using the notation in Section~\ref{subsec:MSIPLC-Ach}, $n=1$, $m=3$, $(i_1,i_2,i_3,i_4,i_5) = (1,3,2,4,5)$, and $(v_1,v_2)=(1,2)$. 
Using Algorithm 2 (selected with probability $1-D/K = 3/5$), and taking $(\omega_1,\omega_2,\omega_3) = (2,1,0)$, $i_{*}=1$, 
\[
\sigma = \begin{pmatrix}1 & 2\\ 2 & 1\end{pmatrix}, \quad \text{and} \quad \pi = \begin{pmatrix} 1 & 2 & 3 & 4 & 5\\ 4 & 2 & 5 & 3 & 1\end{pmatrix},
\] in Step 1 of the proposed IPLC protocol, the user constructs the matrix 
%By applying the single-server IPLT protocol in~\cite{HES2021IndividualJournal} the user constructs the following matrix.    
%\begin{comment}
\begin{equation*}
\mathrm{\mathrm{G}}=
\begin{bmatrix}
0 & 2 & 0 & 1 & 0\\
2 & 0 & 2 & 0 & 1\\
0 & 0 & 2 & 0 & 2\\
\end{bmatrix}.
\end{equation*}
%\end{comment}

%%%%%%%%%%%%%comment
\begin{comment}
\begin{equation*}
\mathrm{\mathrm{G}}=
\begin{bmatrix}
\begin{tikzpicture}
  \matrix (m)[
    matrix of math nodes,
    nodes in empty cells,
    minimum width=width("998888"),
  ] { 
0 & 2 & 0 & 0 & 0\\
0 & 0 & 1 & 2 & 2\\
0 & 0 & 2 & 2 & 0\\
} ;

  \draw (m-3-3.south west) rectangle (m-2-5.north east);\hspace{-0.08cm}
  \draw (m-1-1.south west) rectangle (m-1-2.north east);
\end{tikzpicture}
\end{bmatrix}.
\end{equation*}
\end{comment}
%%%%%%%%%%%%%comment
It is easy to verify that for each of the $2$-subsets $\mathrm{W}_1 = \{1,2\}$, $\mathrm{W}_2 = \{3,4\}$, $\mathrm{W}_3 = \{3,5\}$, and $\mathrm{W}_4 = \{4,5\}$, the row space of $\mathrm{G}$ contains a unique vector (up to scalar multiplication) with support $\mathrm{W}_1$, $\mathrm{W}_2$, $\mathrm{W}_3$, and $\mathrm{W}_4$, respectively, 
e.g., the vectors $\mathrm{U}_1=[0,2,0,1,0]$, 
$\mathrm{U}_2=[0,0,1,0,1]$, 
$\mathrm{U}_3=[1,0,0,0,1]$, and $\mathrm{U}_4=[1,0,2,0,0]$, respectively. 
Let $\mathrm{X} = [X^{\transpose}_1,\dots,X^{\transpose}_5]^{\transpose}$. 
Let $\mathrm{Y} = \mathrm{G}\mathrm{X}$, and let $\mathrm{Y}_1,\mathrm{Y}_2,\mathrm{Y}_3$ denote the rows of $\mathrm{Y}$, 
\begin{align*}
\mathrm{Y}_{1}& =2X_2+X_4,\\
\mathrm{Y}_{2}&=2X_1+2X_3+X_5,\\
\mathrm{Y}_{3}&=2X_3+2X_5.
\end{align*}
Note that $\mathrm{Y}_1,\mathrm{Y}_2,\mathrm{Y}_3$ are linearly independent combinations of the messages $X_1,\dots,X_5$.
Thus, the coded combinations $\mathrm{Z}_1,\dots,\mathrm{Z}_4$ given by
\begin{align*}
\mathrm{Z}_1 & = \mathrm{U}_1\mathrm{X} = 2X_2+X_4, \\    
\mathrm{Z}_2 & = \mathrm{U}_2\mathrm{X} = X_3+X_5, \\
\mathrm{Z}_3 & = \mathrm{U}_3\mathrm{X} = X_1+X_5, \\
\mathrm{Z}_4 & = \mathrm{U}_4\mathrm{X} = X_1+2X_3,  
\end{align*} can be written in terms of the coded messages $\mathrm{Y}_1,\mathrm{Y}_2$ as 
\begin{align*}
\mathrm{Z}_1 & = \mathrm{Y}_1,\\    
\mathrm{Z}_2 & = 2\mathrm{Y}_2,\\
\mathrm{Z}_3 & = 2\mathrm{Y}_2+\mathrm{Y}_3,\\
\mathrm{Z}_4 & = 2\mathrm{Y}_2+2\mathrm{Y}_3.
\end{align*}
Let $\mathrm{C}_1 = [1,0,0]$, $\mathrm{C}_2 = [0,2,0]$, $\mathrm{C}_3 = [0,2,1]$, and $\mathrm{C}_4 = [0,2,2]$ be the coefficient vectors corresponding to $\mathrm{Z}_1$, $\mathrm{Z}_2$, $\mathrm{Z}_3$, and $\mathrm{Z}_4$, respectively, i.e., $\mathrm{Z}_k = \mathrm{C}_k \mathrm{Y}$ for $k\in [4]$. 
Note that the user's demand $\mathrm{Z} =X_1+2X_3= \mathrm{Z}_4$. 

Next, the user sends the matrix $\mathrm{G}$ and the vectors $\mathrm{C}_1,\dots,\mathrm{C}_4$ to each of the servers. 
Note that the user provides the servers with this information so that the servers know the set of all coded combinations among which the user wishes to compute one. 
Then, the user and the servers follow the PLC scheme of~\cite{SJ2018} for $N=2$ servers, $M=4$ coded combinations $\mathrm{Z}_1,\dots,\mathrm{Z}_4$, each of which is a linear combination of $J=3$ (independent) coded messages  $\mathrm{Y}_1,\mathrm{Y}_2,\mathrm{Y}_3$, so that the user can privately recover the coded combination $\mathrm{Z}_4$. 
Note that $\mathrm{Z}_1,\dots,\mathrm{Z}_4$ each contains $T=16$ symbols (from $\mathbb{F}_3$). 
Since $T$ is an integer multiple of $N^M=16$, the scheme of~\cite{SJ2018} is applicable to this setting. 
For all $k\in [4]$ and all $i\in [16]$, let $\mathrm{Z}_k(i)$ denote the $i$th symbol of $\mathrm{Z}_k$, and for all $i\in [16]$, let $(a_i,b_i,c_i,d_i)\triangleq (s_i\mathrm{Z}_1(\tau(i)),s_i\mathrm{Z}_2(\tau(i)),s_i\mathrm{Z}_3(\tau(i)),s_i\mathrm{Z}_4(\tau(i)))$ for a randomly chosen permutation $\tau$ on $[16]$ and a randomly chosen integer $s_i\in \{-1,+1\}$. 
The user's queries generated by the scheme of~\cite{SJ2018} for this example are presented in Table~\ref{tab:ex2}. 
The details are omitted to avoid repetition.

\begin{table}[t]
\caption{The user's queries to the servers for Example 2.} 
\label{tab:ex2}
\centering
\resizebox{0.85\columnwidth}{!}{%
\begin{tabular}{| c | c |}
  \hline 
  \textbf{Server 1} & \textbf{Server 2} \\
  \hline
  $a_1, b_1, c_1, d_1$ & $a_2, b_2, c_2, d_2$ \\
  \hline
  $a_2-d_3$ & $a_1-d_6$ \\
  %\hline
  $b_2-d_4$ & $b_1-d_7$  \\
  $c_2-d_5$ & $c_1-d_8$  \\
  %\hline
  %\hline
  $a_4-b_3$ & $a_7-b_6$  \\
  $a_5-c_3$ & $a_8-c_6$  \\
  $b_5-c_4$ & $b_8-c_7$  \\
  \hline 
  $a_7-b_6+d_9$ & $a_4-b_3+d_{12}$  \\
  $a_8-c_6+d_{10}$ & $a_5-c_3+d_{13}$  \\
  $b_8-c_7+d_{11}$ & $b_5-c_4+d_{14}$  \\
  %\hline
  $a_{11}-b_{10}+c_9$ & $a_{14}-b_{13}+c_{12}$  \\
  \hline
  $a_{14}-b_{13}+c_{12}-d_{15}$ & $a_{11}-b_{10}+c_{9}-d_{16}$  \\
  \hline
\end{tabular}
}
\end{table}

Since any one of $\mathrm{Z}_1,\dots,\mathrm{Z}_4$ can be written as a linear combination of the other three, 
any 3 of the symbols $a_1,b_1,c_1,d_1$ suffice to recover the other one. 
Similarly, any 3 of the symbols $a_2,b_2,c_2,d_2$ suffice to recover the other one. 
Thus, server~1 answers by sending 3 of the symbols $a_1,b_1,c_1,d_1$, say, the symbols $a_1,b_1,c_1$, and 
the 11 coded symbols ${a_2-d_3,b_2-d_4,\dots,a_{14}-b_{13}+c_{12}-d_{15}}$. 
Similarly, server~2's answer consists of 3 of the symbols $a_2,b_2,c_2,d_2$, say, the symbols $a_2,b_2,c_2$, and 
the 11 coded symbols ${a_1-d_6,b_1-d_7,\dots,a_{11}-b_{10}+c_9-d_{16}}$. 
Since the total number of symbols being downloaded from both servers is $28$, the rate of this scheme is $16/28 = 4/7$, 
which matches the upper bound in Lemma~\ref{lem:IPLC-Conv} for ${N=2}$, ${K=5}$, and ${D=2}$, i.e.,  ${(1+1/N+\dots+1/N^{\lceil K/D\rceil-1})^{-1}} = {(1+1/2+1/4)^{-1}} = 4/7$. 

To show that the coded combination $\mathrm{Z}_4$ can be recovered from the answer, it suffices to show that the symbols ${d_1,\dots,d_{16}}$ are recoverable from the answer.  
From Table~\ref{tab:ex1}, it can be seen that the user can readily recover $a_1,b_1,c_1$ and $a_2,b_2,c_2$ from the answer, and hence the user can locally compute $d_1$ and $d_2$ from these two sets of symbols, respectively.
Then, the user can recover ${d_3,\dots,d_8}$ by subtracting off the contribution of ${a_2,b_2,c_2,a_1,b_1,c_1}$ from the coded symbols ${a_2-d_3,b_2-d_4,c_2-d_5,a_1-d_6,b_1-d_7,c_1-d_8}$, respectively. 
Similarly, 
the user can recover $d_9,\dots,d_{14}$ by subtracting off the contribution of the coded symbols ${a_4-b_3,a_5-c_3,b_5-c_4}$ and ${a_7-b_6, a_8-c_6, b_8-c_7}$ from the coded symbols ${a_7-b_6+d_9, a_8-c_6+d_{10},b_8-c_7+d_{11}}$ and ${a_4-b_3+d_12, a_5-c_3+d_{13},b_5-c_4+d_{14}}$, respectively. 
Lastly, $d_{15}$ and $d_{16}$ can be recovered by subtracting off the contribution of $a_{14}-b_{13}+c_{12}$ and $a_{11}-b_{10}+c_9$ from  $a_{14}-b_{13}+c_{12}-d_{15}$ and $a_{11}-b_{10}+c_{9}-d_{16}$, respectively.

\end{example}

\bibliographystyle{IEEEtran}
\bibliography{PIR_PC_Refs}

\end{document}